\documentclass[journal]{IEEEtran}

\usepackage[dvips]{epsfig}
\usepackage[dvips]{graphicx}
\usepackage{boxedminipage}
\usepackage{color} 
\usepackage{amsfonts}
\usepackage{amssymb}
\usepackage{theorem}
\usepackage{amsmath}
\usepackage{graphics,amsmath,amsfonts,epsfig,latexsym,amsfonts,graphicx,amssymb}
\usepackage{setspace}
\usepackage{cite}

\newcommand{\trace}{{\mbox{{\rm {Tr}}}}}
\newcommand{\rank}{{\mbox{\rm{Rank}}}}

\newcommand{\bfalpha}{{\mbox{\boldmath $\alpha$}}}
\newcommand{\bfxi}{{\mbox{\boldmath $\xi$}}}

\newcommand{\bfnu}{{\mbox{\boldmath $\nu$}}}

\newcommand{\bfSigma}{{\mbox{\boldmath $\Sigma$}}}

\newcommand{\bfDelta}{{\mbox{\boldmath $\Delta$}}}

\newcommand{\st}{{\rm s.t.}}
\newcommand{\Prob}{{\rm Prob}}
\newcommand{\bA}{\mathbf{A}}

\newcommand{\bx}{\mathbf{x}}

\newcommand{\bE}{\mathbf{E}}
\newcommand{\bC}{\mathbf{C}}
\newcommand{\bD}{\mathbf{D}}

\newcommand{\bR}{\mathbf{R}}

\newcommand{\bI}{\mathbf{I}}
\newcommand{\bU}{\mathbf{U}}

\newcommand{\bX}{\mathbf{X}}
\newcommand{\bY}{\mathbf{Y}}
\newcommand{\bF}{\mathbf{F}}

\newcommand{\bw}{\mathbf{w}}
\newcommand{\bB}{\mathbf{B}}
\newcommand{\bc}{\mathbf{c}}

\newcommand{\bh}{\mathbf{h}}
\newcommand{\be}{\mathbf{e}}
\newcommand{\ba}{\mathbf{a}}
\newcommand{\bS}{\mathbf{S}}
\newcommand{\br}{\mathbf{r}}
\newcommand{\bg}{\mathbf{g}}
\newcommand{\bG}{\mathbf{G}}

\newcommand{\bhE}{\widehat{\mathbf{E}}}

\newcommand{\btw}{\widetilde{\mathbf{w}}}
\newcommand{\tw}{\widetilde{{w}}}
\newcommand{\bhw}{\widehat{\mathbf{w}}}

\newcommand{\cI}{\mathcal{I}}

\newcommand{\cM}{\mathcal{M}}

\newcommand{\cS}{\mathcal{S}}
\newcommand{\ctS}{\widetilde{\mathcal{S}}}
\newcommand{\cbS}{\bar{\mathcal{S}}}
\newcommand{\cT}{\mathcal{T}}

\newcommand{\btR}{\widetilde{\mathbf{R}}}
\newcommand{\btS}{\widetilde{\mathbf{S}}}

\newcommand{\btF}{\widetilde{\mathbf{F}}}
\newcommand{\btB}{\widetilde{\mathbf{B}}}
\newcommand{\bhR}{\widehat{\mathbf{R}}}

\newtheorem{prop}{Proposition}
\newtheorem{thm}{Theorem}

\theorembodyfont{\rmfamily}

\newtheorem{rmk}{Claim}

\begin{document}
\title{Joint User Grouping and Linear Virtual Beamforming: Complexity, Algorithms
and Approximation Bounds\thanks{M. Hong, M. Razaviyayn and Z.-Q.
Luo's research is supported in part by the US AFOSR, grant number
00029548, by NSF, Grant No. CCF-1216858, and by a research gift from
Huawei Technologies Inc. Z. Xu is supported by the China NSF under
the grant 11101261 and partly by a grant of ``The First-class
Discipline of Universities in Shanghai".}}
\author{Mingyi Hong, Zi Xu, Meisam Razaviyayn  and Zhi-Quan Luo
\thanks{M. Hong, M. Razaviyayn and Z.-Q. Luo are with the Department of Electrical and Computer
Engineering University of Minnesota, Minneapolis, USA (email:
\{mhong, razav002, luozq\}@umn.edu). Z. Xu is with the Department of
Mathematics, Shanghai University, Shanghai, China (Correponding
author, email: xuzi@shu.edu.cn).}} \maketitle

{
\begin{abstract}
In a wireless system with a large number of distributed nodes, the
quality of communication can be greatly improved by pooling the
nodes to perform joint transmission/reception. In this paper, we
consider the problem of optimally selecting a subset of nodes from
potentially a large number of candidates to form a virtual
multi-antenna system, while at the same time designing their joint
linear transmission strategies. We focus on two specific application
scenarios: {\it 1)} multiple single antenna transmitters
cooperatively transmit to a receiver; {\it 2)} a single transmitter
transmits to a receiver with the help of a number of cooperative
relays. We formulate the joint node selection and beamforming
problems  as {\it cardinality constrained optimization problems}
with both discrete variables (used for selecting cooperative nodes)
and continuous variables (used for designing beamformers). For each
application scenario, we first characterize the computational
complexity of the joint optimization problem, and then propose novel
semi-definite relaxation  (SDR) techniques to obtain approximate
solutions. We show that the new SDR algorithms  have a guaranteed
approximation performance in terms of the gap to global optimality,
regardless of channel realizations.
 The effectiveness of the proposed algorithms is demonstrated
via numerical
experiments. 
\end{abstract}}


\section{Introduction}\label{secIntroduction}

With the proliferation of rich multimedia services as well as smart
mobile devices, the demand for wireless data has been increasing
explosively in recent years. To accommodate the growing demand for
wireless data, practical techniques that can significantly improve
the spectrum efficiency of existing wireless systems must be
developed. In this paper, we focus on a combination of two such
techniques: partial node cooperation and collaborative beamforming.

In a cellular network, cooperation can be achieved by allowing the
neighboring base stations (BSs) to form a virtual multi-antenna
system for joint transmission and reception, a scheme known as
cooperative multipoint {(CoMP)} \cite{Foschini06}. It can
effectively cancel the inter-BS interference, and has been included
into the next-generation wireless standards such as {3GPP Long Term
Evolution-Advanced (LTE-A)}; see e.g., \cite{gesbert10,Foschini06,
Irmer11}. For example, in a downlink network, assuming the users'
data signals are known at all BSs, then either the capacity
achieving non-linear dirty-paper coding (DPC) (see, e.g.,
\cite{Caire03,yu07perantenna}), or simpler linear precoding schemes
such as zero-forcing (ZF) (see, e.g., \cite{Spencer04, zhang09,
zhang10JSAC}) can be used for joint transmission. In addition,
various cooperation schemes have been proposed to exploit spatial
diversity among the mobile users as well \cite{Laneman03,
Sendonaris03,Nassab08, Dehkordy09, Nam08, Zhao06}. In these schemes,
users assist each other in relaying information to the desired
destinations by using various strategies such as
amplify-and-forward (AF) and decode-and-forward (DF). 

However, the cost of cooperation can outweigh its benefit when the
size of the cooperation group grows large. Such costs include the
overhead incurred by exchanging control and data signals among the
cooperating nodes (either via backhaul networks or air interfaces);
it can also include efforts required to maintain
system level synchronization \cite{Love08,gesbert10, Irmer11,huang13tsp}. 
To control the size of cooperation group, various {{\it partial}
cooperation schemes} have been developed recently. In the setting of
a cellular network, partial cooperation among the BSs amounts to
judiciously clustering the BSs into (possibly overlapping) small
cooperation groups, within which they cooperatively transmit to or
receive from the users \cite{zhang09, Ng10, Papadogiannis10}. In
\cite{hong12sparse, kim12, zeng10}, joint BS clustering and
beamforming problems are formulated as certain {\it sparse}
beamformer design problems, in which the sparsity of the virtual
beamformer corresponds to the size of the cooperation groups.
Partial cooperation in the relay networks has also been studied
recently. In \cite{Ibrahim08, Zhao06}, the authors propose to select
a {\it single} relay (out of many candidates) so that certain
performance metric at the receiver is optimized. Alternatively,
references \cite{Zhao06, Jing09, yu12relay, Michalopoulos06} study
the {\it multiple} relay selection problem. In particular, the
authors of \cite{Jing09} propose to increase the number of relays
until adding an additional one decreases the received SNR. Reference
\cite{Michalopoulos06} formulates the relay selection problem as a
Knapsack problem \cite{garey79}, and proposes greedy algorithms for
this problem. However, these schemes generally assume simplified
underlying cooperation schemes after fixing the cooperative set. For
example, references \cite{zhang09,Papadogiannis10} use simple zero
forcing strategies for intra-cluster transmission, while references
\cite{Jing09,yu12relay} assume that the cooperative relays transmit
with full power. There has been no performance analysis for these
partial cooperation schemes. This is due to the {\it mixed-integer}
nature of the problem when treating the group membership (which is a
set of discrete variables) as optimization variables.

{In this paper, we study the problem of optimally partitioning the
transmit nodes into cooperation groups, while at the same time
designing their cooperation strategies. We focus on two related
network settings in which either multiple nodes cooperatively
transmit to a receiver, or a single node transmits to the receiver
with the help of a set of cooperative relays. In both cases, the
cooperative nodes are allowed to form a virtual antenna system, and
they can jointly design the virtual transmit beamformers. More
specifically, our objective is to find a subset of cooperative nodes
(with given cardinality) and their joint linear beamformers so that
the system performance measured by the receive signal to noise ratio
(SNR) is maximized. We formulate the problem as a cardinality
constrained quadratic program and study its computational
complexity. Furthermore, we develop novel semi-definite relaxation
(SDR) algorithms  for this mixed integer quadratic program and prove
that they  have a guaranteed approximation performance in terms of
the gap to global optimality, regardless of channel realization.
Compared to the existing SDR algorithms and their analysis
\cite{luo10SDPMagazine,Luo07approximationbounds,
Nemirovski_Roos_Terlaky_1999, zhang:optimal11} which focus on
quadratic problems with continuous variables, our work deals with
mixed-integer cardinality constrained quadratic optimization
problems and
therefore has a significantly broader scope. } 

The rest of the paper is organized as follows. In Section
\ref{secSystem}, we introduce the virtual beamforming problem
without node grouping. Section \ref{secAdmissionControl}--Section
\ref{secRelay} consider the joint node grouping and virtual
beamforming problems in various settings. Section
\ref{secSimulation} presents numerical results. The concluding
remarks and future works are given in Section \ref{secConclusion}.


{\it Notations}: For a symmetric matrix $\mathbf{X}$,
$\mathbf{X}\succeq 0$ signifies that $\mathbf{X}$ is positive
semi-definite. We use $\trace[\mathbf{X}]$, $\mathbf{X}^H$, and
$\lambda_i(\bX)$ to denote the trace, Hermitian transpose, and the
$i$-th largest eigenvalue of $\mathbf{X}$, respectively. The
notation ${\rm diag}(\bX)$ denotes a matrix consisting the diagonal
value of $\bX$. For an index set $\cS$, the notations $\bX[i,j]$ and
$\bX[\cS]$ denote the $(i,j)$-th element of $\bX$ and the principal
submatrix of $\bX$ indexed by the set $\cS$, respectively.
Similarly, we use $\bx[i]$ and $\bx[\cS]$ to denote the $i$-th
element of a vector $\bx$ and the subvector of $\bx$ with the
elements in set $\cS$, respectively. For a complex scalar $x$, the
notation $\bar{x}$ denotes its complex conjugate . Let
$\delta_{i,j}$ denote the Kronecker function, which takes the value
$1$ if $i=j$, and $0$ otherwise. The notations $\mathbf{I}_n$,
$\be$ and $\be_i$ denote, respectively, the $n\times n$ identity
matrix, the
all one vector in $\mathbb{R}^n$, and the $i$-th unit vector in $\mathbb{R}^n$.  
Some other notations are listed in Table~\ref{tableSymbols}.

\vspace*{-0.1cm}
\begin{table*}[htb]
\caption{\small {A List of Notations} }\vspace*{-0.5cm}
\begin{center}
{\small
\begin{tabular}{|c |c | c|c|}
\hline
$\mathcal{M}$& The set of all transmit nodes & $N$& The number of antennas at the receiver\\
\hline
$\bR$& The channel covariance matrix  & $\bw$& The virtual beamformer\\
\hline
$P$& The transmit power budget for each node & $Q$& The size of cooperative group\\
\hline
$\bx$& The vector of both discrete and continuous variables & $\cS$& The support of a beamformer $\bw$\\
\hline
$\btw$& The dimension-reduced beamformer $\bw[\cS]$ & $\bY$& The rank-1 matrix $\btw\btw^H$\\
\hline
$\bX$& The rank-1 matrix $\bx\bx^H$& $L$& The sample size for randomization\\
 \hline
 \hline
 $\mathbf{f}$ & The first hop channel& $\bg$ &The second hop channel\\
 \hline
$P_0$& The transmit power for the transmitter & $\bfnu$& The noise at the relay\\
 \hline
 $\bS$& The Channel matrix $P_0\mathbb{E}[(\bg\odot\mathbf{f})(\bg\odot\mathbf{f})^H]$& $\bF$& The channel matrix $\mathbb{E}[(\bg\odot\bfnu)(\bg\odot\bfnu)^H]$\\
 \hline
\end{tabular} } \label{tableSymbols}
\end{center}
\vspace*{-0.2cm}
\end{table*}

{\section{Virtual Beamforming with Full
Cooperation}\label{secSystem}}
\subsection{A Single-hop Network}\label{subSingleHop}

Let us first describe the virtual beamforming (VB) problem with all
transmit nodes fully cooperating in transmission. Suppose there is a
set $\cM=\{1,\cdots, M\}$ of transmitters each equipped with a
single antenna, and there is a single receiver with $N\ge 1$ receive
antennas. This setting depicts for instance an uplink cellular
network, where the transmit nodes are the users and the receive node
is the BS. We are interested in the case $M> N$, where the receiver
cannot cancel the interference among the transmit nodes if they
transmit simultaneously and independently. In this case, the benefit
of transmit nodes cooperation in improving system performance is
more pronounced. Let $h_{i,n}\in\mathbb{C}$ denote the channel
between transmitter $i$ and the $n$-th antenna of the receiver, and
define $\bh_{n}\triangleq[h_{1,n},\cdots, h_{M,n}]^H$. Suppose only
second order statistics on the channels are available, that is, both
the transmitters and the receiver only know
$\mathbb{E}[\bh_n\bh_n^H]=\bR_n\succ 0$, for $n=1,\cdots, N$. Let
$z\sim\mathcal{CN}(0,1)$ denote the (normalized) noise at the
receiver.

For tractability, we restrict ourselves to a simple transmit
cooperation strategy in which the cooperative transmitters can form
a {\it linear} virtual beamformer for joint transmission. Let
$w_i\in\mathbb{C}$ denote the complex antenna gain for {transmitter}
$i$, which satisfies an individual transmit power constraint:
$|w_i|^2\le P$. Define $\bw\triangleq[w_1,\cdots, w_M]^H$. When all
the transmitters participate simultaneously in the beamforming, they
transmit {\it the same} data signal to the
receiver by using {\it distinct} antenna gains. 
The transmit nodes can share their data signals by the following
steps: {\it 1)} identify the node whose data is to be transmitted;
{\it 2)} the identified node broadcasts its data to all nearby
nodes, who subsequently decode the data.


%


{Assume that the receiver performs spatially matched
filtering/maximum ratio combining (which is equivalent to the MMSE
receiver in this case), then the total received signal power can be
expressed as:
$\sum_{n=1}^{N}|\bh^H_n\bw|^2=\bw^H(\sum_{n=1}^{N}\bh_n\bh^H_n)\bw$.
Let ${\bR}\triangleq\sum_{n=1}^{N}\bR_n$, and assume that the noise
power is normalized to $1$, then the averaged signal to noise ratio
(SNR) at the receiver is given by: {${\rm
SNR}=\mathbb{E}\bigg[\sum_{n=1}^{N}|\bw^H\bh_n|^2\bigg]=\bw^H\bR\bw.$}
To optimize the averaged SNR at the BS, one can solve the following
quadratic program (QP)}{
\begin{align}
\max_{\bw}&\quad \bw^H\bR\bw\label{problemOriginal} \\
{\st}&\quad |w_i|^2\le P, \ i=1,\cdots ,M\nonumber.
\end{align}}
\hspace{-0.2cm} At this point, it may appear that solving the above
SNR maximization problem with per-antenna power constraints can be
done easily, using for example algorithms based on uplink-downlink
duality proposed in \cite{yu07perantenna}. Unfortunately, as will be
seen later in Section \ref{subAdmissionControlSystem}, this
seemingly simple problem turns out to be computationally very
difficult, for {\it general} channel covariance matrix $\bR$ with
rank larger than one. {In fact, the uplink-downlink duality theory
developed in \cite{yu07perantenna} and related works depends
critically on the assumption that $\bR$ is of rank one. Later in
Claim \ref{remarkComplexity} we will see that indeed in our case
when $\rank(\bR)=1$, solving problem \eqref{problemOriginal} is
easy.

It is worth mentioning that the formulation \eqref{problemOriginal}
is equally applicable to solving the sum SNR maximization problem
when the {\it instantaneous} channel states $\{\bh_n\}$ are
available. In this case, $\bR$ should be replaced by the
instantaneous channel
$\bhR\triangleq\sum_{n=1}^{N}\bh_{n}\bh^H_{n}$. 


\begin{figure*}[htb]\vspace{-0.2cm}
    \begin{minipage}[t]{0.49\linewidth}
    \centering
    {\includegraphics[width=
0.75\linewidth]{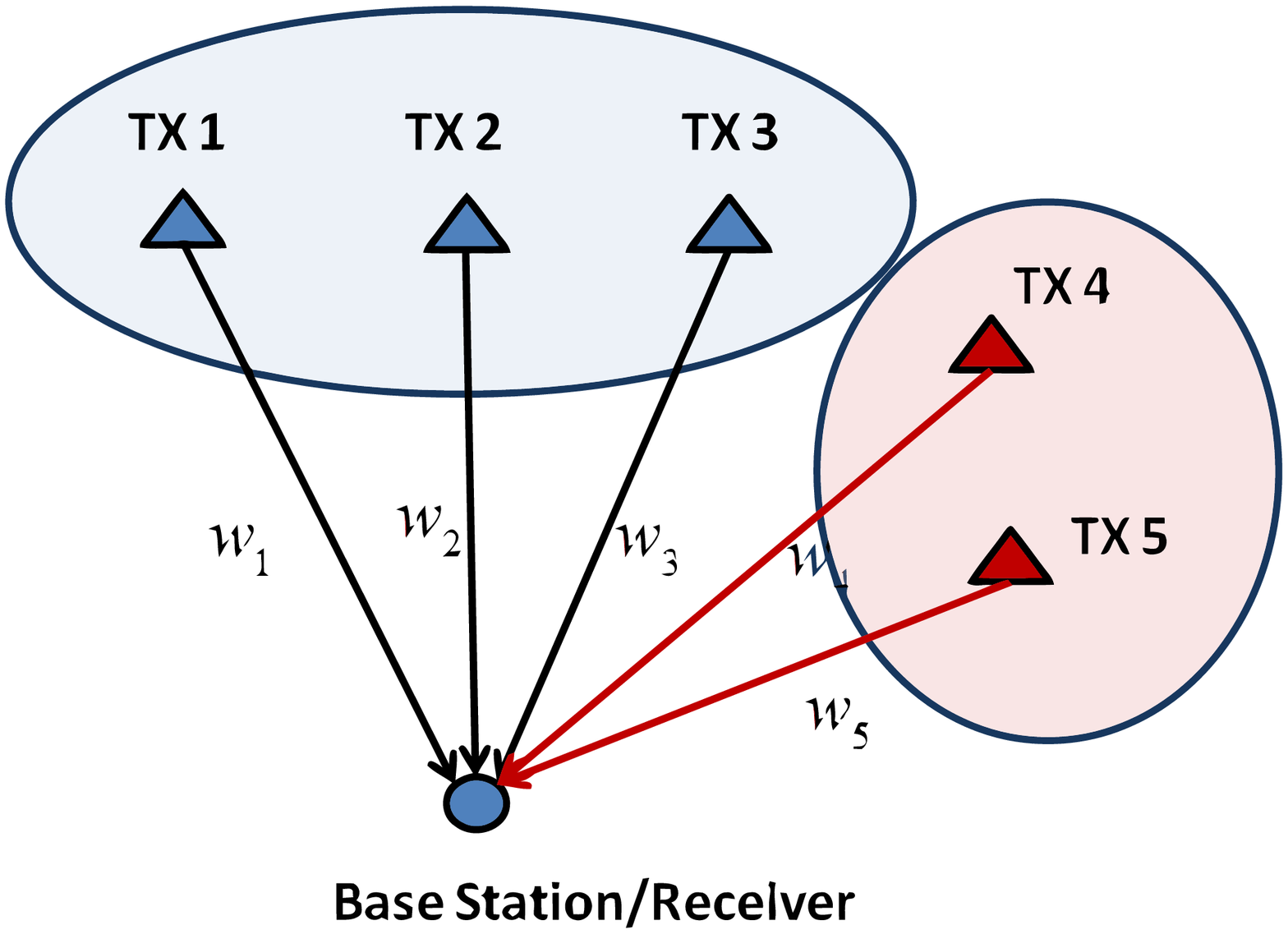} \caption{\small
Illustration of joint node grouping and VB. Nodes $1$-$3$ and $4$,
$5$ are divided into two different groups.}\label{figGrouping} }
\vspace{-0.5cm}
    \end{minipage}
    \hfill
        \begin{minipage}[t]{0.49\linewidth}
       {\includegraphics[width=
0.75\linewidth]{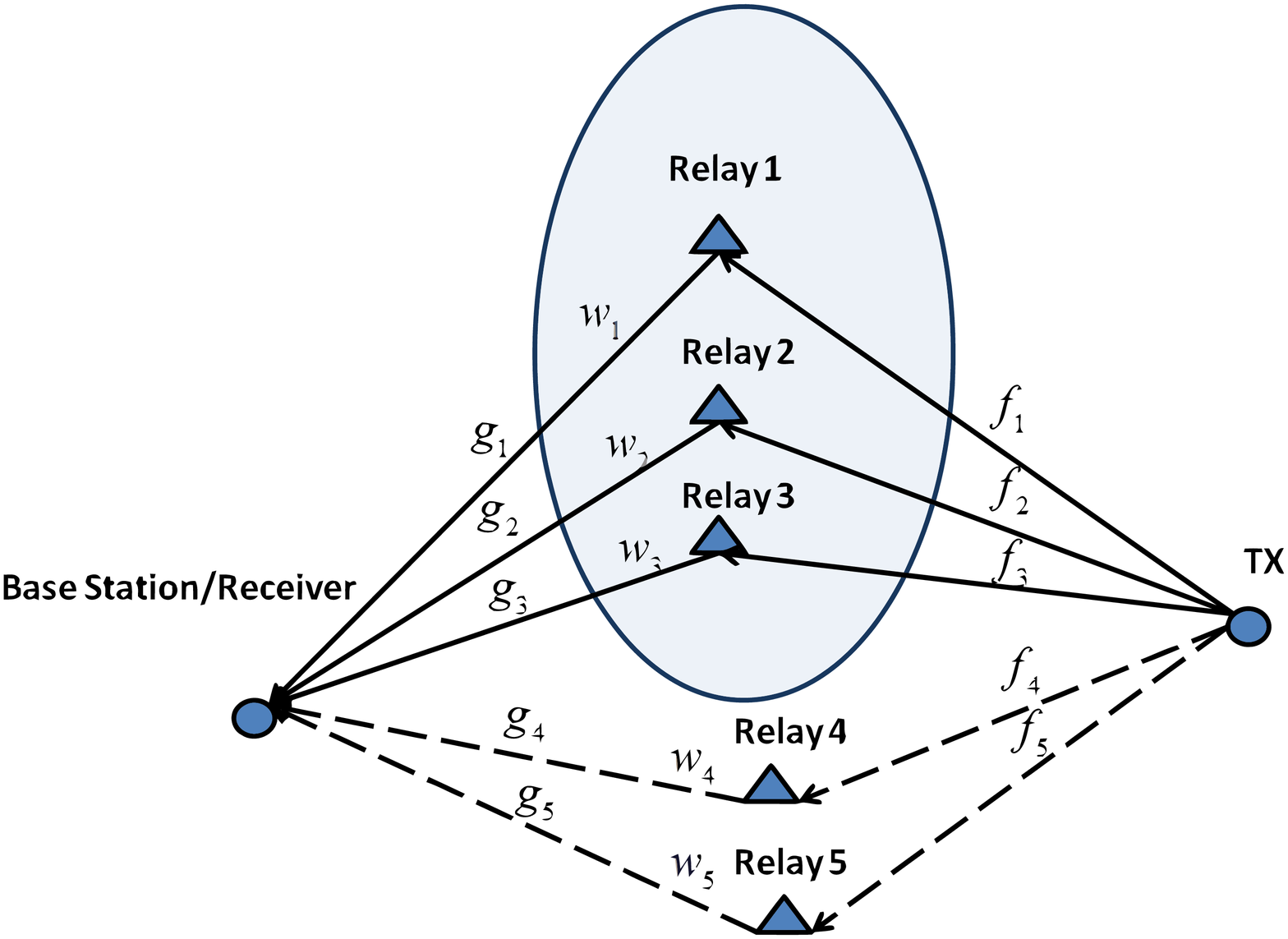} \caption{\small Illustration
of joint relay selection and VB. Relays $1$-$3$ are the serving
relays. }\label{figRelay} }\vspace{-0.2cm}
\end{minipage}
\end{figure*}

%

\subsection{A Two-hop Network}\label{subTwoHop}
In the previous single hop model, it is assumed that the hop
connecting the source and the cooperative nodes is {\it reliable},
in the sense that all the cooperative nodes can perfectly decode the
signals to be jointly transmitted. 
Alternatively, when the quality of the first hop communication also
needs to be taken into consideration, the problem can be formulated
in the context of the a two-hop relay network, as we explain below.

Consider a network with a pair of transceiver and a set of $M$
relays, each of which has a single antenna (see Fig. \ref{figRelay}
for an illustration). Assume that there is no direct link between
the transmitter and the receiver. Let $\{f_i\}_{i=1}^{M}$ and
$\{g_i\}_{i=1}^{M}$ denote the complex channel coefficients between
the transmitter and the relays, and between the relays and the
receiver, respectively. We focus on a popular AF relay protocol, in
which the transmitter broadcasts the desired signal to the relays,
who subsequently forward the signals to the receiver. Assume that
there is a large number of relays available, and any group of them
can form a virtual multi-antenna system for transmission. 

Let us use $s\in\mathbb{C}$ to denote the message transmitted by the
transmitter; use $P_0$ to denote the transmit power; use $\nu_i$ to
denote the noise at the $i$-th relay with power $\sigma_{\nu}^2$.
Then the signal $x_i$ received at the $i$-th relay can be expressed
as $x_i=\sqrt{P_0}f_i s+\nu_i$. Again use $w_i$ to denote the
complex gain applied by the $i$-th relay, which satisfies the power
constraint $|w_i|^2\le P$. It follows that the transmitted signal of
$i$-th relay is given by $y_i=w_i x_i$. Using this expression, the
averaged transmit power of relay $i$ can be expressed as{
\begin{align}
\mathbb{E}[|y_i|^2]=|w_i|^2\mathbb{E}[x_i\bar{x}_i]=|w_i|^2\left(P_0
\mathbb{E}[|f_i|^2]+\sigma_{\nu}^2\right).\nonumber
\end{align}}
Let $n\in\mathcal{CN}(0,\sigma_n^2)$ denote the noise at the
receiver, then the received signal is given by {\begin{align}
z&=\sum_{i=1}^{M}g_iy_i+n=\underbrace{\sqrt{P_0}\sum_{i=1}^{M}w_if_ig_i
s}_{\rm signal} +\underbrace{\sum_{i=1}^{M} w_i g_i \nu_i +n}_{\rm
noise}.\label{eqSignalReceived}
\end{align}}
The averaged signal power at the receiver is then given by
{\begin{align}
\mathbb{E}\bigg[\big|\sqrt{P_0}\sum_{i=1}^{M}w_if_ig_i
s\big|^2\bigg]=\bw^H\bS\bw
\end{align}}
{\hspace{-0.1cm}}where {$\bS\triangleq
P_0\mathbb{E}[(\mathbf{f}\odot\bg)(\mathbf{f}\odot\bg)^H]$}, with
$\odot$ denoting the componentwise product.  When
$\{\nu_i\}_{i=1}^{M}$ and $\{g_i\}_{i=1}^{M}$ are independent from
each other, the averaged noise power is given by
\cite{Nassab08}{\begin{align} \mathbb{E}\left[\big|\sum_{i=1}^{M}
w_i g_i \nu_i +n\big|^2\right]=\bw^H\bF\bw+\sigma^2_n,
\end{align}}
{\hspace{-0.2cm}} where
$\bF\triangleq\mathbb{E}[(\bg\odot\bfnu)(\bg\odot\bfnu)^H]$.
Additionally, if we further assume that the noises
$\{\nu_i\}_{i=1}^{M}$ are independent, then $\bF$ becomes diagonal:
{$\bF\triangleq\sigma_{\nu}^2{\rm diag}(\mathbb{E}[\bg\bg^H])$}. It
follows that the averaged SNR at the receiver is given
by\cite{Nassab08}: $ {\rm
SNR}=\frac{\bw^H\bS\bw}{\sigma^2_n+\bw^H\bF\bw}\label{eqSNRRelay}$.
To optimize the averaged SNR at the receiver, the following problem
needs to be solved {
\begin{align}
\max_{\bw}&\quad\frac{\bw^H\bS\bw}{\sigma^2_n+\bw^H\bF\bw}\label{problemRelaySimple}\\
{\st}&\quad |w_i|^2\left(P_0
\mathbb{E}[|f_i|^2]+\sigma_{\nu}^2\right)\le P, \ i=1,\cdots, M.
\nonumber
\end{align}}
We remark that when the set of per-relay power constraints is
replaced by a single sum-power constraint, the above problem is
equivalent to a principal generalized eigenvector problem, which is
easily solvable \cite{Nassab08}. However, as will be explained in
more detail in Section \ref{secRelay}, when the per-relay power
constraints are present, \eqref{problemRelaySimple} turns out to be
computationally difficult.

In practice, when the number of transmit/relay nodes becomes large,
allowing all of them to cooperate at the same time induces heavy
signaling overhead (related to nodes' exchange of data and control
signals) and computational efforts (related to computing the optimal
virtual beamformer for all the nodes) \cite{Love08}. 
To address these issues, it is necessary to divide the
transmit/relay nodes into different cooperative groups while at the
same time optimizing their virtual beamformers. How to do so in
either single-hop or two-hop networks will be the focus of the rest
of this paper.


\section{Joint Admission Control and VB}\label{secAdmissionControl}
In this section, we consider a basic setting in which the aim is to
find a {\it single} cooperative group with a fixed size. Such
admission control problem is important as fixing the size of the
group can effectively control the cooperation and computational
overhead. {Although admission control for wireless networks is a
well-studied subject (see \cite{Matskani08,Matskani09,
liu13deflation}), existing solutions cannot be directly applied in
our setting because they are designed for conventional wireless
networks {\it without} node cooperation.} \vspace{-0.5cm}

\subsection{Problem Formulation and Complexity Status} \label{subAdmissionControlSystem}
Let $Q$  denote the desired size of the cooperation group, and
introduce the set of binary variables $a_i\in\{0,1\}, i\in\cM$ to
indicate the transmit nodes' group membership: when $a_i=1$, node
$i$ is being assigned to the cooperation group. Let
$\ba\triangleq[a_1,\cdots, a_M]$. Then the joint admission control
and VB problem is given as the following cardinality constrained
program {
\begin{align}
v_1^{\rm CP}=\max_{\bw, \mathbf{a}}&\quad \bw^H\bR\bw\label{problemCCQP}\tag{CP1} \\
{\st}&\quad |w_i|^2\le a_i P, \ i=1,\cdots ,M\nonumber \\
&\quad \sum_{i=1}^{M}a_i=Q,\ \quad a_i\in\{0,1\}, \ i=1,\cdots,
M\nonumber.
\end{align}}
\hspace{-0.2cm}Note that $a_i=0$ implies $|w_i|^2=0$, that is, node
$i$ does not transmit. In the following, we will use $v_1^{\rm
CP}(\bw)$ to indicate the objective value achieved by a feasible
solution $\bw$.

%

{In order to express the problem in a simpler form, we introduce a
homogenizing variable $\gamma\in\{-1,1\}$ and change the domain of
the discrete variables to $\{-1,1\}$. By doing so problem
\eqref{problemCCQP} can be equivalently reformulated in the
following quadratic form:  } {
\begin{align}
\max_{\bw, \ba, \gamma}&\quad \bw^H\bR\bw\nonumber\\
{\st}&\quad \bw^H\be_i\be^T_i\bw+\frac{P}{4}(a_i-\gamma)^2\le P, \ i=1,\cdots ,M\nonumber \\
&\quad \sum_{i=1}^{M}(a_i+\gamma)^2=4Q,\nonumber\\
&\quad a_i\in\{-1,1\}, \ i=1,\cdots, M,\ \gamma\in\{-1,1\}\nonumber.
\end{align}}
\hspace{-0.2cm}{{To see the equivalence, we first perform a change
of variable domain by defining: $\hat{a}_i=2{a}_i-1$, for all $i$,
where $a_i$ is the original variable with the domain $\{0,1\}$. Then
we split each $\hat{a}_i$ by $\hat{a}_i=\gamma \tilde{a}_i$ for a
new variable $\tilde{a}_i\in\{-1,1\}$. By doing so the constraints
can be shown to be quadratic in both $\tilde{\ba}$ and $\gamma$. For
notational simplicity, below we still use $a_i$ to denote the new
variable $\tilde{a}_i\in\{-1,1\}$.}}

{ After such transformation, we see that $\gamma a_i=-1$ implies
$w_i= 0$, i.e., node $i$ does not join the cooperative group. To
further express the problem in a standard quadratic form of both the
binary and continuous variables, we need the following
definitions}{\small
\begin{subequations}
\begin{align}
&\bC_{i,0}\triangleq\frac{1}{4}\left(\be_i\be^T_i+\be_{M+1}\be^T_{M+1}
-\be_i\be^T_{M+1}-\be_{M+1}\be^T_{i}\right)\in\mathbb{R}^{(M+1)\times
(M+1)},\\
&\bC_{i,1}\triangleq\be_i\be^T_i\frac{1}{P}\in\mathbb{R}^{M\times M}\label{eqC2}\\
&\bD_i\triangleq{\rm blkdg}[\bC_{i,0}, \bC_{i,1}]\in\mathbb{C}^{(2M+1)\times(2M+1)}\label{eqD}\\
&\widetilde{\bR}\triangleq{\rm blkdg}[{\bf 0}, \bR]\in\mathbb{C}^{(2M+1)\times(2M+1)}\label{eqbtR}\\
&\bB_0\triangleq\left[ \begin{array}{ll}
\bI& \be \\
\be^T & M
\end{array}\right]\in\mathbb{R}^{(M+1)\times(M+1)},\\
&\bB={\rm blkdg}[\bB_0,{\bf 0}]\in\mathbb{R}^{(2M+1)\times(2M+1)}\label{eqB}\\
& \bx\triangleq[\ba^T, \gamma, \bw^T]^T\ \bx_0\triangleq [\ba^T,
\gamma]^T, \bx_1\triangleq\bw.\label{eqx}
\end{align}
\end{subequations}}
\hspace{-0.2cm} {We can now compactly write \eqref{problemCCQP} as a
quadratic problem of the newly defined vector $\bx$, which contains
both binary and continuous variables:} {
\begin{subequations}
\begin{align}
\max_{\bx}&\quad\bx^H\btR\bx\label{problemReformulate}\tag{R1}\\
{\st }&\quad \bx^H\bD_i\bx\le 1, \ i=1\cdots, M\label{eqPowerConstraintQP}\\
&\quad \bx^H\bB\bx=4Q,\label{eqCardinatliyConstraintQP}\\
&\quad \bx[i]\in\{-1,1\}, \ i=1,\cdots, M+1\nonumber.
\end{align}
\end{subequations}}
\hspace{-0.2cm} { We emphasize again that in this new notation,
$\bx[i]\bx[M+1]=-1$ implies that node $i$ is not in the cooperative
group (i.e., $\bw[i]=0$), or equivalently $\bx[M+1+i]= 0$ from
definition \eqref{eqx}.}

Towards finding a solution for problem
\eqref{problemCCQP}/\eqref{problemReformulate}, the first task is to
analyze their computational complexity. Our analysis, to be
presented shortly, shows that these problems are difficult even when
fixing the values of the binary variables $\{a_i\}$.

Let $\cS\subset\cM$ with $|\cS|=Q$ denote the {\it support} of a
feasible solution $\bw$ to problem \eqref{problemCCQP}:
$\cS=\{i:a_i=1\}$. When $\cS$ is fixed, problem
\eqref{problemCCQP} is equivalent to 
the following QP {
\begin{align}
\max_{\btw\in\mathbb{C}^Q}&\quad\btw^H\bR[\cS]\btw\label{problemQP}\tag{QP1}\\
{\rm s.t.} &\quad|\tw_i|^2\le P, i=1,\cdots,Q\nonumber.
\end{align}}
In the following we analyze the computational complexity of
 \eqref{problemQP}.

\begin{prop}\label{propositionNPHardReal}
Solving the problem \eqref{problemQP} is NP-hard in the number of
transmit nodes.
\end{prop}
\begin{proof}
{We only prove the claim with real variables. The complex case can
be derived similarly. The claim is proved by a polynomial time
reduction from a known NP-complete problem called equal partition
problem \cite{garey79}, which can be described as follows. Given a
vector $\bc$ consisting of positive integers $c_1, \cdots, c_Q$, the
equal partition problem decides whether there exists a subset $\cI$
such that{
\begin{align}
\frac{1}{2}\sum_{i=1}^{Q}c_i=\sum_{i\in\cI}c_i\label{eqEqualPartition}.
\end{align}}}
\hspace{-0.2cm} {In the following, we will show that a special case
of problem \eqref{problemQP} is equivalent to an instance of equal
partition problem.} Suppose $\bc^T\bc=C>0$. Let
$\bR[\cS]=(-\bc\bc^T+2C\bI_{Q})\succ 0$. The claim is that the
problem \eqref{problemQP} has the maximum value of $2CQP$ if and
only if there exists a set $\cI$ satisfying
\eqref{eqEqualPartition}. The objective of problem \eqref{problemQP}
can be written as{
\begin{align}
\btw^T\bR[\cS]\btw&=-|\btw^T\bc|^2+2C\sum_{i=1}^{Q}|\tw_i|^2\nonumber\\
&\le 2C\sum_{i=1}^{Q}|\tw_i|^2\le2CQP, {\rm ~whenever~} |\tw_i|^2\le
P.
\end{align}}
\hspace{-0.2cm}Consequently, the maximum value for problem
\eqref{problemQP} is $2CQP$ if and only if $-|\btw^T\bc|^2=0
{\rm~and~} |\tw_i|^2=P.$ This is equivalent to the existence of an
index set $\cI$ such that \eqref{eqEqualPartition} is true.
\end{proof}
{It is important to note that when $\cS=\cM$, \eqref{problemQP} is
the same as the VB problem \eqref{problemOriginal}. Therefore we can
readily conclude that solving problem
\eqref{problemOriginal} is also NP-hard.} 

\begin{rmk}\label{remarkComplexity}
{\it When $\bR$ admits certain special structures, both
\eqref{problemQP} and \eqref{problemCCQP} may be easy to solve. One
such example is that when $\rank(\bR)=1$, which corresponds to the
special case where the receiver has a single antenna, and the {\it
instantaneous} SNR is considered. Another relevant case is that when
$\bR$ is a diagonal matrix, which happens when all the transmit
nodes' channels are {\it independent and zero mean}. For both cases,
problems \eqref{problemQP} and \eqref{problemCCQP} are separable
among the variables, and their solutions can be easily obtained in
closed form.}
\end{rmk}

\subsection{The Semi-definite Relaxation}
Our  proposed algorithm is based on the technique called
semi-definite relaxation (SDR), which has been widely used to solve
problems in communications and signal processing
\cite{luo10SDPMagazine}. We emphasize that unlike conventional SDR
methods, in which the problems to be relaxed have either all
continuous (e.g., \cite{Luo07approximationbounds}) or all discrete
(e.g., \cite{Goemans:1995}) variables, our problem
\eqref{problemCCQP}/\eqref{problemReformulate} is of {\it mixed
integer} nature. Consequently our algorithm and analysis to be
presented differ significantly from those developed in the existing
literature.

We first introduce two semi-definite programs (SDPs) which are
relaxations of problems \eqref{problemReformulate} and
\eqref{problemQP}. Define a variable $\bX\triangleq\bx\bx^H$. Define
two index sets $\cI\triangleq\{1,\cdots, M+1\}$ and
$\bar{\cI}\triangleq\{M+2, \cdots, 2M+1\}$. Then
$\bX_0\triangleq\bX[\cI]$ and $\bX_1\triangleq\bX[\bar{\cI}]$ denote
the leading and trailing principal submatrices of $\bX$,
respectively. Clearly $\bX_0=\bx_0\bx^T_0$ and $\bX_1=\bx_1\bx_1^H$.
Moreover, we have $\rank(\bX)=1$ and $\bX_{0}[i,j]\in\{-1,1\}$ for
all $i,j\in\cI$. The following SDP is a relaxation of
\eqref{problemReformulate}, by removing the non-convex constraint
$\rank(\bX)=1$ and by replacing $\bX_{0}[i,j]\in\{-1,1\}$ by
$\bX_{0}[i,j]\in[-1,1]$, for all $i,j\in\cI$: {\begin{subequations}
\begin{align}
v_1^{\rm SDP}=\max_{\bX\succeq 0}&\quad\trace[\btR\bX]\label{problemSDP}\tag{SDP1} \\
{\st}&\quad\trace[\bD_i\bX]\le 1, \ i=1\cdots, M\label{eqPowerConstraintSDP}\\
&\quad \trace[\bB\bX]=4Q \label{eqCardinalityConstraintSDP}\\
&\quad\bX[i,i]=1,i=1,\cdots, M+1\label{eqConditionDiscreteDiagonal}.
\end{align}
\end{subequations}}
\hspace{-0.1cm}Note that we did not explicitly include the
conditions $\bX_{0}[i,j]\in[-1,1]$, for all $i,j\in\cI$, as it can
be ensured by the set of conditions
\eqref{eqConditionDiscreteDiagonal} and $\bX\succeq 0$. As the above
problem is a {\it relaxation} of problem \eqref{problemReformulate},
we must have $v_1^{\rm SDP}\ge v_1^{\rm CP}$. Denote the optimal
solution for this problem as $\bX^*$. Since  all the data matrices
$\bB$, $\bD_i$ and $\btR$ are block diagonal matrices, removing the
off-diagonal blocks of an optimal solution does not change either
its optimality or feasibility.  Thus, without loss of generality we
can assume $\bX^*={\rm blkdg}[\bX^{*}_0, \bX^{*}_1]$.

Similarly, let $\bY\triangleq\btw\btw^H\in\mathbb{C}^{Q\times Q}$.
The following problem is a relaxation of the problem
\eqref{problemQP}, for a given index set $\cS\subseteq\cM$ with
$|\cS|=Q${\begin{subequations}
\begin{align}
\max_{\bY\succeq 0}&\quad\trace[\bR[\cS]\bY]\label{problemSDP-QCQP}\\
{\st}&\quad\bY[i,i]\le P,\ i=1,\cdots,
Q.\label{eqPowerConstraintSDP-QCQP}
\end{align}
\end{subequations}}
Let us denote the optimal solution of this problem by $\bY^*$.

{ Fig. \ref{figFlowChart} below shows the relationship among
different problem formulations introduced so far. For problems
\eqref{problemSDP} and \eqref{problemSDP-QCQP}, the following claims
summarize some useful properties of their optimal solutions. These
properties will be used later for analyzing the quality of certain
approximate solutions for the original problem
\eqref{problemCCQP}/\eqref{problemReformulate}.}

   \begin{figure}[ht]
    \centering
     {\includegraphics[width=
0.8\linewidth]{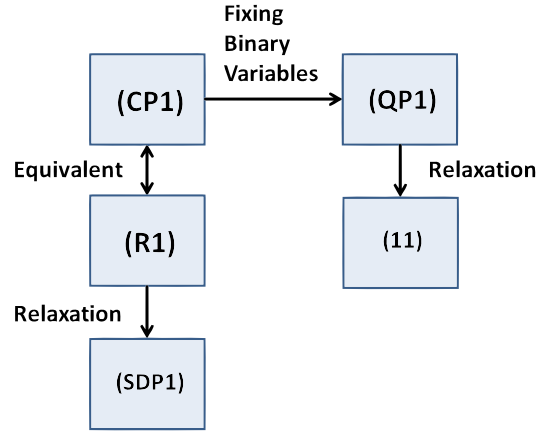} \caption{\footnotesize
Relationship among different problem
formulations.}\label{figFlowChart} }\vspace*{-0.2cm}
    \end{figure}

\begin{rmk}\label{remarkTight}
{\it At optimality, the set of constraints
\eqref{eqPowerConstraintSDP} and \eqref{eqPowerConstraintSDP-QCQP}
must be all tight. That is {
\begin{align}
\frac{1}{P}\bX^{*}_1[i,i]&=\frac{1}{2}+\frac{1}{2}\bX^{*}_0[i,M+1],\
\forall~i=1,\cdots,M\label{eqPowerConstraintTight}, \\
\bY^*[i,i]&=P,\ \forall~i=1,\cdots,Q.
\end{align}}}
\end{rmk}
\vspace{-0.3cm}

%

\begin{rmk}\label{remarkXLowerBound}
{\it The sum of the last column of $\bX^{*}_0$ admits a closed form
expression: {$ \sum_{i=1}^{M}\bX^{*}_0[i,M+1]=2Q-M.$}}
\end{rmk}

{Claim \ref{remarkTight} can be shown straightforwardly using a
contradiction argument. Claim \ref{remarkXLowerBound} can be derived
using the cardinality constraint \eqref{eqCardinalityConstraintSDP}.
Due to space limitations, we refer the readers to
\cite{hong12vmimo_proof} for a formal proof. }



\subsection{The Proposed
Algorithm}\label{subAlgorithmAdmissionConstrol}

In this section, we propose a randomized algorithm that generates an
approximate solution for problem \eqref{problemCCQP}. To highlight
ideas, we list below the main steps of the algorithm:
\begin{enumerate}
\item  Compute the optimal solution $\bX^*$ of the relaxed problem
\eqref{problemSDP};

\item Determine the discrete variables $\bx_0$ and the set $\cS$ according to
$\bX^{*}_0$;

\item Fixing $\cS$, compute the optimal solution $\bY^*$ of problem
\eqref{problemSDP-QCQP};

\item Randomly generate a sample of feasible $\bw$'s using
$\bY^{*}$;

\item Select the solution that achieves the best objective
value for problem \eqref{problemCCQP}.
\end{enumerate}


{Intuitively, steps 1)--2) select the set of cooperative nodes,
while the rest of the steps determine the virtual beamformer among
the selected nodes. To formally describe the proposed algorithm, the
following definitions are needed. Let $\cS\subseteq\cM$ be an index
set, and let $\bY^*$ denote the corresponding solution for problem
\eqref{problemSDP-QCQP}. Let us factorize $\bY^{*}$ as
$\bY^{*}=\bfDelta^H\bfDelta$. Then define{
\begin{align}
\bE_i&\triangleq\bfDelta\bC_{i,1}[\cS]\bfDelta^H, \
\bE\triangleq\bfDelta\bR[\cS]\bfDelta^H\nonumber.
\end{align}}
\hspace{-0.2cm} Let us further decompose $\bE$ as {\small$
\bE=\bU\bfSigma\bU^H\label{eqEDecomposition}$}. Then the diagonal
matrix $\bfSigma$ can be expressed as{
\begin{align}
\bfSigma=\bU^H\bE\bU=\bU^H\bfDelta\bR[\cS]\bfDelta^H\bU\label{eqSigma}.
\end{align}}}
\hspace{-0.2cm}
Let $L$ denote the sample size of the randomization, and let the
superscript $(l)$ denote the index of a random sample. Let
$\underline{\bx}_{(Q)}$ and $\underline{\br}_{(Q)}$ respectively
denote the $Q$-th largest value in the sets
$\{\bX^{*}_0[i,M+1]\}_{i=1}^{M}$ and $\{\bR[i,i]\}_{i=1}^{M}$. The
proposed algorithm is described in Table \ref{tableAlgorithm}. 


\begin{table}[htb]
\begin{center}
\vspace{-0.1cm} \caption{The Proposed Algorithm for Admission
Control} \label{tableAlgorithm} {\small
\begin{tabular}{|l|}
\hline
\\
S1: Compute the solution $\bX^*$ of problem
\eqref{problemSDP}\\

S2: Find a set $\cT$ of indices such that $|\cT|=Q$ and \\
\quad\quad\quad $\cT=\{j:\bX^{*}_0[j,M+1]\ge \underline{\bx}_{(Q)}\}$;\\
\quad \quad S2a: {\bf If}\quad\quad
$\trace\left[\bR[\cT]\bX^{*}_1[\cT]\right]\ge\frac{QP}{M}\trace[\bR]$, let $\cS=\cT$;\\
\quad \quad S2b: {\bf Else} ~~Let $\cS=\{j:\bR[j,j]\ge \underline{\br}_{(Q)}\}$;\\
\quad\quad Let $\bar{\cS}=\cM\setminus \cS$;\\

S3: Set $\bx_0[M+1]=1$ and  $\bx_0[j]=1$, for all $j\in\cS$; \\
\quad~ Set $\bx_0[i]=-1$ for all $i\in\bar{\cS}$;\\

S4: Compute the solution $\bY^*$ of problem
\eqref{problemSDP-QCQP} with index set $\cS$;\\

{\bf For} $\ell=1,\cdots, L$\\

S5: Generate $\bfxi^{(\ell)}\in\{-1,1\}^{Q}$ by randomly and independently \\
\quad ~ generating its components from $\{-1,1\}$;\\

S6: Compute $t^{(\ell)}=\sqrt{\max_{i\in\cS}(\bfxi^{(\ell)})^T\bU^H\bE_i\bU\bfxi^{(\ell)}}$;\\

S7: Compute
$\btw^{(\ell)}=\frac{1}{t^{(\ell)}}\bfDelta^H\bU\bfxi^{(\ell)}$; \\
\quad ~~Let $\bw^{(\ell)}[\cS]=\btw^{(\ell)}$ and
$\bw^{(\ell)}[\bar{\cS}]=\bf{0}$;\\
{\bf End For}\\

S8: Compute $\ell^*=\arg_{\ell}\max(\bw^{(\ell)})^H\bR\bw^{(\ell)}$;
let $\bw^*=\bw^{(\ell^*)}$;

\\
  \hline
\end{tabular}}
\vspace{-0.3cm}
\end{center}
\end{table}

This algorithm can be viewed as a generalization of the algorithm
developed by Nemirovski et al. and Zhang et al.
\cite{Nemirovski_Roos_Terlaky_1999, zhang:optimal11} for
approximating continuous quadratic programs. The novelty of this
algorithm lies in steps S2)--S3), in which discrete variables are
determined. {Below we motivate Step S2). Without rank relaxation,
$\bX_0=\bx_0\bx^H_0$ is the block variable representing the discrete
variables, so it is reasonable to select cooperative groups using
the elements of this matrix. Recall that in problem
\eqref{problemReformulate}, a node $i$ joins the cooperative group
if $\bx_0[i]\bx_0[M+1]=1$, which, combined with the definition
$\bX_0=\bx_0\bx^T_0$, implies that $\bX_0[i,M+1]=1$. Ideally, we
should form the cooperative group by choosing $Q$ elements in the
set $\ctS=\{i: \bX^*_0[i,M+1]=1, i\in\cM\}$. However it is possible
that $|\ctS|<Q$, as we have relaxed $\bX_0[i,j]\in\{-1,1\}$ to
$\bX_0[i,j]\in[-1,1]$. As a result, we instead choose the largest
$Q$ elements in the set $\{\bX_0^*[i,M+1]\}_{i=1}^{M}$. {Using Claim
\ref{remarkXLowerBound}, it is seen that there is a lower bound on
the sum of such $Q$ elements: $\sum_{j\in\cT}\bX^*_0[j,M+1]\ge
\frac{(2Q-M)Q}{M}$. This bound will be instrumental in the following
performance analysis.} Additionally, steps S2a)--S2b) are some
technical refinement of the selection procedure that are needed
later for the proof of the approximation bounds.}


{The reason for using steps S4)--S8) to generate the solution
$\bw^*$ is twofold: {\it 1)} the optimal objective value $v_1^{\rm
CP}(\bw^{*})$ can be written down analytically; {\it 2)} $\bw^*$ is
always feasible. See the following two claims for more details
regarding these two properties. Formal arguments for these claims
are relegated to Appendix \ref{appendixRemarkFeasibility}.}

\begin{rmk}\label{remarkObjective}
{\it The objective value of the problem \eqref{problemCCQP}
evaluated at a solution $\bw^{(\ell)}$ is given by{
\begin{align}
v_1^{\rm
CP}(\bw^{(\ell)})&=\frac{1}{(t^{(\ell)})^2}\trace[\bR[\cS]\bY^*]\label{eqObjectiveQP}.
\end{align}}
\hspace{-0.2cm} This result implies that $v_1^{\rm
CP}(\bw^{*})=\frac{1}{\min_{\ell}(t^{(\ell)})^2}\trace[\bR[\cS]\bY^*]$.}
\end{rmk}

\begin{rmk}\label{remarkFeasibility}
{\it For all $l=1,\cdots, L$, the  solution
$\bx^{(\ell)}\triangleq[\bx^H_0, (\bw^{(\ell)})^H]^H$ is a feasible
solution for the problem \eqref{problemReformulate}. Moreover,
$\bw^{(\ell)}$ is a feasible solution to \eqref{problemCCQP}.}
\end{rmk}

\subsection{The Analysis of the Quality of the
Solution}\label{subApproximationRatio}

{Clearly the solution $\bw^*$ generated by the proposed algorithm is
only a feasible solution for \eqref{problemCCQP}. A natural question
then is: how good this solution is in terms of the achieved receive
SNR. In the following, we will show that the quality of $\bw^*$ can
be indeed guaranteed. That is, compared with the globally optimal
objective $v_1^{\rm CP}$, there is a {\it finite} constant
$\alpha_1>1$ such that: {$ v_1^{\rm CP}(\bw^*)\ge
\frac{1}{\alpha_1}v_1^{\rm CP}$.} The constant $\alpha_1$ is
referred to as the {\it approximation ratio} of the solution
$\bw^*$. The smaller the value of $\alpha_1$, the better the quality
of the solution $\bw^*$. The following result provides a finite data
independent bound for $\alpha_1$. The proof is relegated to Appendix
\ref{appProofTheorem1}.}

\begin{thm}\label{thmAppRatio}
If $\bw^*$ is generated using the algorithm in Table
\ref{tableAlgorithm}, then with high probability, we have $v_1^{\rm
CP}(\bw^*)\ge \frac{1}{\alpha_1}v_1^{\rm CP}$, with $\alpha_1$
bounded above by{
\begin{align}
\alpha_1\le\frac{8M\lambda_{1}(\bR)}{\sum_{i=1}^{M}\lambda_{i}(\bR)}\ln(5Q)<
8M\ln(5Q)\label{eqAppRatio1}.
\end{align}}
\end{thm}

{It is interesting to see that for any channel realization,
$\alpha_1$ is finite. Moreover, when the eigenvalues of $\bR$ are
roughly uniformly distributed, the derived bound is of the order
$\mathcal{O}(\ln(Q))$, which is better than the case where
$\bR$ has a single dominant eigenvalue. 
Nevertheless, it is important to note that the theoretical
approximation ratio obtained above characterizes the quality of the
{\it worst} solutions. { It implies that, compared to the global
optimal solution or the cardinality constrained problem
\eqref{problemCCQP}/\eqref{problemReformulate},
 the solution generated by the SDR approach
cannot be {\it arbitrarily bad}  {\it regardless problem instance}}.
As we will see later in our numerical results, the practical
performance of the algorithm is much better  than the derived
worst-case bound \eqref{eqAppRatio1}.}
%
%
%
%
%
%
%
%

\section{Joint Transmit Node Scheduling and VB}\label{secScheduling}

{The previous section considers the case where a subset of nodes are
selected for transmission. Such a scheme may not be fair to all the
nodes, as the ones that are being excluded from the cooperative set
do not get served. In this section we study a  generalized
formulation that provides fairness among the transmit nodes. }

\subsection{Problem Formulation and Complexity Status}
{Suppose there are {\it two} orthogonal time slots available for
transmission. The problem is to effectively schedule $Q$ transmit
nodes to the first slot and the rest $M-Q$ nodes to the second one,
in a way that the {\it minimum} SNR among these two time slots is
maximized. In this case, effectively there are two virtual
transmitters in the network, and the scheduling scheme promotes
fairness among the virtual transmitters.

Let $\bw_k\in\mathbb{C}^{M}$ denote the virtual beamformer used in
the $k$-th time slot, and let $w_{k,i}\in\mathbb{C}$ denote node
$i$'s antenna gain in slot $k$. Suppose that in both time slots the
channel matrices $\bR$ remains the same. Mathematically, the problem
is given by}{
\begin{align}
v_2^{\rm CP}=&\max_{\{\bw_1,\bw_2, \bfalpha\}}\min_{k=1,2}\quad \bw_k^H\bR\bw_k\label{problemCCQPScheduling}\tag{CP2} \\
{\st}&\ |w_{1,i}|^2\le a_i P, \ |w_{2,i}|^2\le (1-a_i) P, \ i=1,\cdots ,M\nonumber \\
&\quad \sum_{i=1}^{M}a_i=Q,\ \quad a_i\in\{0,1\}, \ i=1,\cdots,
M.\nonumber
\end{align}}
\hspace{-0.2cm}In the following, we will use $v^{\rm
CP}_2(\bw_1,\bw_2)$ to denote the objective achieved by a feasible
tuple $(\bw_1,\bw_2)$. {Note that it is possible to extend
\eqref{problemCCQPScheduling} to the  multiple time slot case by
using more discrete variables ($M$ discrete variables per slot).
However, the resulting analysis will become quite involved. In the
remainder of this paper, we will consider the 2-slot case only.}

Similar to the case of \eqref{problemReformulate}, let us introduce
a homogenizing variable $\ell\in\{-1,1\}$. Let us define $\bB_0$,
$\bC_{i,0}$ and $\bC_{i,1}$ the same way as in
\eqref{eqC2}--\eqref{eqx}. Let us further define{\small
\begin{align}
\tilde{\bC}_{i,0}&\triangleq\frac{1}{4}\left(\be_i\be^T_i+\be_{M+1}\be^T_{M+1}
+\be_i\be^T_{M+1}+\be_{M+1}\be^T_{i}\right)\in\mathbb{R}^{(M+1)\times(M+1)}\nonumber\\
{\bA}_{i,1}&\triangleq{\rm blkdg}[\bC_{i,0}, \bC_{i,1}, {\bf{0}}]\in\mathbb{C}^{(3M+1)\times(3M+1)}, \nonumber\\
{\bA}_{i,2}&\triangleq{\rm blkdg}[\tilde{\bC}_{i,0}, {\bf{0}},\bC_{i,1}]\in\mathbb{C}^{(3M+1)\times(3M+1)}\nonumber\\
\btB&\triangleq{\rm blkdg}[\bB_0,{\bf{0}}, {\bf{0}}]\in\mathbb{R}^{(3M+1)\times (3M+1)}\nonumber\\
\bx&\triangleq[\ba^T, \ell, \bw^T_1, \bw^T_2]^T, \
\bx_0\triangleq[\ba^T, \ell], \ \bx_1\triangleq\bw_1, \
\bx_2\triangleq\bw_2.\nonumber
\end{align}}
\hspace{-0.2cm}Then problem \eqref{problemCCQPScheduling} can be
equivalently written as{
\begin{align}
\max_{\bx}&\min_{k=1,2}\quad\bx_k^H\bR\bx_k\label{problemReformulateScheduling}\tag{R2}\\
{\st }&\quad \bx^H\bA_{i,1}\bx\le 1, \ i=1\cdots, M\nonumber\\
&\quad \bx^H{\bA}_{i,2}\bx\le 1, \ i=1\cdots, M\nonumber\\
&\quad \bx^H\btB\bx=4Q, \ \quad \bx[i]\in\{-1,1\}, \ i=1,\cdots,
M+1\nonumber.
\end{align}}
The max-min scheduling problem is at least as difficult as its
admission control counterpart, as when fixing the group membership,
the subproblem of maximizing the per-group SNR is the same as
\eqref{problemQP}. To see this, we again fix an index set
$\cS_1\subset\cM$ with $|\cS_1|=Q$, and let ${\cS}_2=\cM\setminus
\cS_1$. Then the problem \eqref{problemCCQPScheduling} reduces to
two QPs, one for each slot $k$:{
\begin{align}
\max_{\btw_1\in\mathbb{C}^{|\cS_k|}}&\quad\btw_k^H\bR[\cS_k]\btw_k\label{problemQP2}\\
{\rm s.t.} &\quad|\btw_{k}[i]|^2\le P, \ i=1,\cdots,|\cS_k|\nonumber
\end{align}}
\hspace{-0.2cm} {One may observe that each of these problems has the
same structure as problem \eqref{problemQP}. It follows from
Proposition \ref{propositionNPHardReal} that solving either one of
them is difficult for general $\bR$. Interestingly, unlike the
admission control problem, the scheduling problem is difficult even
when $\bR$ is diagonal or is of rank 1. The following result
summarizes the complexity status, the proof of which can be found in
\cite{hong12vmimo_proof}.

\begin{prop}\label{propositionNPHardRealScheduling}
Solving problem \eqref{problemCCQPScheduling} is strongly NP-hard
for general channel matrix $\bR$, as well as for the special cases
when $\bR$ is either rank 1 or diagonal.
\end{prop}}

\vspace{-0.3cm}
\subsection{The Proposed Algorithm}\label{subAlgorithmScheduling}
{The scheduling algorithm we propose below is similar to the one for
the admission control problem---we use the solutions of a relaxation
of \eqref{problemCCQPScheduling} to construct approximate
solutions.} To proceed, define
$\bX_0\in\mathbb{R}^{(M+1)\times(M+1)}$,
$\bX_1,\bX_2\in\mathbb{R}^{M\times M}$, and let $\bX={\rm
blkdg}[\bX_0,\bX_1,\bX_2]$. The SDR of problem
\eqref{problemReformulateScheduling} is given
by{\small\begin{subequations}
\begin{align}
{v}_2^{\rm SDP}=\max_{\bX\succeq 0}\min_{k=1,2}&\quad\trace[\bR\bX_k]\label{problemSDPScheduling}\tag{SDP2} \\
{\st}&\quad\trace[\bA_{i,1}\bX]\le 1, \ i=1\cdots, M\label{eqPowerConstraintSDPScheduling1}\\
&\quad\trace[{\bA}_{i,2}\bX]\le 1, \ i=1\cdots, M\label{eqPowerConstraintSDPScheduling2}\\
&\quad \trace[\btB\bX]=4Q \label{eqCardinalityConstraintSDPscheudling}\\
&\quad\bX[i,i]=1,i=1,\cdots, M+1.\nonumber
\end{align}
\end{subequations}}
\hspace{-0.2cm}Similarly, for a fixed index set $\cS_k$, the SDR of
problem \eqref{problemQP2} is{\begin{subequations}
\begin{align}
\max_{\bY_k\succeq 0}&\quad\trace[\bR[\cS_k]\bY_k]\label{problemSDP-QCQP2}\\
{\st}&\quad\bY_k[i,i]\le P,\ i=1,\cdots,
|\cS_k|.\label{eqPowerConstraintSDP-QCQP2}
\end{align}
\end{subequations}}
%


{To formally describe the proposed algorithm, we need to introduce a
few definitions that are similar to those in Section
\ref{subAlgorithmAdmissionConstrol}. Let $\cS_k\subseteq\cM$ be any
index set, and let $\bY^*_k$ denote the corresponding solution for
problem \eqref{problemSDP-QCQP2}. Decompose $\bY_k^{*}$ by
$\bY_k^{*}=\bfDelta_k^H\bfDelta_k$, for $k=1,2$. Define the
following{
\begin{align}
\bE_{i,k}&\triangleq\bfDelta_k\bC_{i,1}[\cS_k]\bfDelta_k^H, \
\bE_k\triangleq\bfDelta_k\bR[\cS_k]\bfDelta_k^H\nonumber.
\end{align}}
Let us decompose $\bE_k$ using its eigendecomposition: $
\bE_k=\bU_k\bfSigma_k\bU_k^H$.}

{The proposed algorithm for joint scheduling and VB follows almost
identical steps of the admission control algorithm in Table
\ref{tableAlgorithm}, with only minor changes. Below we list the
main steps of the proposed algorithm.
\begin{enumerate}
\item Compute the optimal solution $\bX^{*}$ of problem \eqref{problemSDPScheduling},
such that all the constraints in
\eqref{eqPowerConstraintSDPScheduling1} and
\eqref{eqPowerConstraintSDPScheduling2} are tight \footnote{To find
a solution required by step 1), we can start by any optimal solution
$\bX^*$ of \eqref{problemSDPScheduling}, and increase the diagonal
elements of $\bX_1^*$ or $\bX^*_2$ until all the constraints in
\eqref{eqPowerConstraintSDPScheduling1} and
\eqref{eqPowerConstraintSDPScheduling2} are satisfied.}.

\item Find the set $\cS_1$ with $|\cS_1|=Q$ by
$\cS_1=\{j:\bX^{*}_0[j,M+1]\ge \underline{\bx}_{(Q)}\}$; Set
${\cS}_2=\cM\setminus \cS_1$.

\item For $k=1,2$, compute the solution $\bY_k^*$ of problem
\eqref{problemQP2} with index set $\cS_k$.

\item Perform twice the randomization steps identical to those in Step
5)--Step 8) in Table \ref{tableAlgorithm}, replacing $\{\bU,
\bfDelta, \bE, \bE_{i}\}$ with $\{\bU_k, \bfDelta_k, \bE_k,
\bE_{i,k}\}$, $k=1,2$; obtain samples
$\{\bw^{(\ell)}_1,\bw^{(\ell)}_2\}_{\ell=1}^{L}$.
\item Select the best sample by
$\ell^*=\arg\max_{\ell}\min_{k=1,2}\{(\bw_k^{(\ell)})^H\bR\bw_k^{(\ell)}\}$.
\end{enumerate}

Let us pause to discuss the differences between the above SDR
algorithm and its counterpart for admission control. After deciding
the set $\cS_1$ and $\cS_2$, {\it two} separate randomization
procedures are needed, one for each set $\cS_1$ and ${\cS_2}$.
Intuitively, after deciding $\cS_1$ and $\cS_2$, we have completed
the scheduling task. What remains to be done is to perform VB for
the nodes allocated to each slot. This is the goal of the
randomization procedure. Moreover, the best sample $\ell^*$ is
selected according to the max-min SNR criteria, which promotes
fairness among the two virtual beamformers. }

Using the argument identical to that presented in Claim
\ref{remarkObjective} and Claim \ref{remarkFeasibility}, we can
verify that for all $\ell$, $[\bx_0,\bw^{(\ell)}_1, \bw_2^{(\ell)}$]
must be feasible for problem \eqref{problemReformulateScheduling}.
Moreover, the optimal value $v_2^{\rm CP}(\bw^*)$ can be expressed
in closed form{
\begin{align}
(\bw_k^{*})^H\bR\bw_k^{*}
&=\frac{1}{\min_\ell(t_k^{(\ell)})^2}\trace\left[\bR[\cS]\bY^{*}_k\right],\
k=1,2\nonumber.
\end{align}\vspace{-0.5cm}}

%
%

Let $\alpha_2$ be the approximation ratio for a solution $(\bw^*_1,
\bw^*_2)$, defined as ${v}_2^{\rm
CP}(\bw^*_1,\bw^*_2)\ge\frac{1}{\alpha_2}{v}_2^{\rm CP}$. Using
techniques similar to the proof of Theorem \ref{thmAppRatio}, one
can show that $\alpha_2$ can be bounded by{
\begin{align}
\alpha_2\le\frac{8M\lambda_1(\bR)}{\min\{Q,M-Q\}\lambda_M(\bR)}\ln(12\max\{Q,
M-Q\})\label{eqAppRatioSchedule}.
\end{align}}
However, compared to Theorem \ref{thmAppRatio}, the above bound is
less powerful since it is {\it dependent} on the channel
realization. The proof of this result can be found in
\cite{hong12vmimo_proof}

\section{Joint Relay Grouping and VB}\label{secRelay}

In this section we show that the approaches developed in the
previous sections are also applicable in relay networks. The problem
here is to select a set of relays to form a virtual multi-antenna
system, and at the same time design their virtual beamformers for
signal relaying.

\subsection{Problem Formulation and Complexity Status}

Suppose  $Q$ out of $M$ relays are to be selected for transmission.
Using the system model described in Section \ref{secSystem}, we can
formulate the joint relay grouping and beamforming problem as
follows:
 {
\begin{align}
v^{\rm CP}_3=\max_{\bw, \ba}&\quad\frac{\bw^H\bS\bw}{\sigma^2_n+\bw^H\bF\bw}\label{problemRelay}\tag{CP3}\\
{\st}&\quad |w_i|^2\left(P_0
\mathbb{E}[|f_i|^2]+\sigma_{\nu}^2\right)\le a_i P, \ i=1,\cdots, M,
\nonumber\\
&\quad \sum_{i=1}^{M}a_i=Q, \ a_i\in\{0,1\}, \ i=1,\cdots,
M.\nonumber
\end{align}}
where the objective is the receive SNR.
This problem can be compactly written as{
\begin{subequations}
\begin{align}
v_3^{\rm QP}=\max_{\bx}&\quad \frac{\bx^H\btS\bx}{\sigma^2_n+\bx^H\btF\bx}\label{problemRelayReformulate}\tag{R3}\\
{\rm \st}&\quad \bx^H\bD_i\bG_i\bx\le 1, \ i=1,\cdots, M,
\label{eqRelayConstraint1}\\
&\quad \bx^H\bB\bx=4Q, \ \bx[i]\in\{-1,1\}, \ i=1,\cdots,
M+1.\label{eqRelayConstraint2}
\end{align}
\end{subequations}}
{\hspace{-0.2cm}}where $\bD_i$, $\bB$ and $\bx$ are given in
\eqref{eqC2}--\eqref{eqx}, and the $(2M+1)\times (2M+1)$ matrices
$\btF$, $\btS$ and $\bG_i$ are defined as {
\begin{subequations}
\begin{align}
{\btS}&\triangleq\left[ \begin{array}{ll}
\bf{0}& \bf{0} \\
\bf{0}&\bS
\end{array}\right],\ {\btF}\triangleq\left[ \begin{array}{ll}
\bf{0}& \bf{0} \\
\bf{0}&\bF
\end{array}\right], \nonumber\\
\bG_i&\triangleq \left[ \begin{array}{ll}
\bf{I}_{M+1}& \bf{0} \\
\bf{0}&\be_i\be^T_i{\left(P_0
\mathbb{E}[|f_i|^2]+\sigma_{\nu}^2\right)}
\end{array}\right], \ i=1,\cdots, M\nonumber.
\end{align}
\end{subequations}}

As always, we first analyze the computational complexity of joint
relay grouping and VB problem \eqref{problemRelay}. The following
theorem shows that solving this problem is generally NP-hard. We
refer the readers to Appendix \ref{appNPhardRelay} for proof
details.
\begin{prop}\label{thmNPHardRelay}
Solving \eqref{problemRelay} is NP-hard in general.
\end{prop}

{It is worth noting that problem~\eqref{problemRelay} is easy when
there is no correlation between the channels, or equivalently when
both $\mathbf{S}$ and $\mathbf{F}$ are diagonal. The reason is that
for a fixed value of $t\ge 0$, solving the
following feasibility problem is easy 
{
\begin{align}
&\frac{\bw^H \bS \bw}{ \sigma_n^2 + \bw^H \bF\bw} \geq t\nonumber\\
& |w_i|^2 \left(P_0 \mathbb{E}(|f_i|^2) + \sigma_v^2\right) \leq a_i P,\; \forall i\nonumber\\
& \sum_{i=1}^{M} a_i = Q, \quad a_i \in \{0,1\}.\nonumber
\end{align}}
{\hspace{-0.2cm}}By performing a bisection on $t$, we can obtain the
optimal solution for \eqref{problemRelay}.}

Similar to problem \eqref{problemCCQPScheduling}, one can consider
the relay scheduling problem over two time slots. Mathematically,
this problem can be formulated as{
\begin{align}
\max_{\{\bw_1,\bw_2, \bfalpha\}}\min_{k=1,2}&\quad \frac{\bw_k^H \bS \bw_k}{ \sigma_n^2 + \bw_k^H \bF\bw_k}\label{CP4}\tag{CP4} \\
{\st}&\quad |w_{k,i}|^2\left(P_0 \mathbb{E}[|f_i|^2] + \sigma_v^2\right)\le a_i P,  \ i=1,\cdots ,M,\ k=1,2\nonumber \\
&\quad \sum_{i=1}^{M}a_i=Q,\ \quad a_i\in\{0,1\}, \ i=1,\cdots,
M\nonumber.
\end{align}}
\hspace{-0.1cm}It turns out that this problem is NP-hard even for
 diagonal channel matrices (see \cite{hong12vmimo_proof} for the proof).

\begin{prop}\label{thmNPHardRelayDiagonal}
For diagonal channel matrices, problem~\eqref{CP4} is NP-hard.
\end{prop}

\subsection{The  SDR Algorithm}\label{subAlgorithmRelay}
Again let us define $\bX\triangleq\bx\bx^H$. The SDR of the
reformulated problem \eqref{problemRelayReformulate} is given by{
\begin{subequations}
\begin{align}
v_3^{\rm SDP}\triangleq\max_{\bX\succeq 0}&\quad \frac{\trace[\btS\bX]}{\sigma^2_n+\trace[\btF\bX]}\label{problemSDPRelay}\tag{SDP3}\\
{\rm \st}&\quad \trace[\bD_i\bG_i\bX]\le 1, \ i=1,\cdots, M,
\label{eqRelaySDPConstraint1}\\
&\quad \trace[\bB\bX]=4Q, \label{eqRelaySDPConstraint2}\\
&\quad  \bX[i,i]=1, \ i=1,\cdots, M+1.\label{eqRelaySDPConstraint3}
\end{align}
\end{subequations}}
{\hspace{-0.1cm}}Let  $\bX^*$ denote the optimal solution of this
problem. Clearly, we must have ${\trace[\btS\bX^*]}=v_3^{\rm
SDP}(\sigma^2_n+\trace[\btF\bX^*])$. Moreover, $\bX^*$ must be the
optimal solution of the following problem, with an optimal objective
value $v_3^{\rm SDP}\sigma^2_n${
\begin{align}
\max_{\bX\succeq 0}&\quad \trace[(\btS-v_3^{\rm SDP}\btF)\bX]\label{problemSDPRelayReformulate}\tag{SDP4}\\
{\rm \st}&\quad
\eqref{eqRelaySDPConstraint1}-\eqref{eqRelaySDPConstraint3}\nonumber.
\end{align}}
This problem is a relaxation of the following QP{
\begin{align}
\max_{\bx}&\quad \bx^H(\btS-v^{\rm SDP}_3\btF)\bx\nonumber\\
{\rm \st}&\quad
\eqref{eqRelayConstraint1}-\eqref{eqRelayConstraint2}\nonumber.
\end{align}}
The similarity between problem \eqref{problemSDPRelayReformulate}
and the relaxed node selection problem \eqref{problemSDP} suggests a
natural two-step approach to obtain a feasible solution $\bx^*$ of
\eqref{problemRelayReformulate}:
\begin{enumerate}
\item Solve problem \eqref{problemSDPRelay}, obtain $v_3^{\rm SDP}$.

\item Obtain $\bx^*=[\bx^*_0 ,\bw^*]$ by applying the algorithm in Table \ref{tableAlgorithm}, with the matrices $\btR$, $\bD_i$ replaced
 by $\btS-v_3^{\rm SDP}\btF$, $\bD_i\bG_i$, respectively.
\end{enumerate}
Using the same argument given in Claim \ref{remarkFeasibility}, we
can show that the resulting vector $\bx^*$ must be feasible for our
relay selection problem.


Unfortunately, at this point we are still unable to derive a finite
(data independent) approximation bound for this SDR algorithm. The
main difficulty is that, unlike the case of \eqref{problemSDP}, the
coefficient matrix $\btS-v_3^{\rm SDP}\btF$ in the objective of
\eqref{problemSDPRelayReformulate} is no longer a positive definite
matrix (in fact it is indefinite). This implies that some  key
properties (e.g., Claim \ref{remarkTight}) no longer hold true in
this case. Nevertheless, in our numerical experiments, we did
observe that the proposed SDR algorithm generates high quality
approximate solutions for problem \eqref{problemRelay}.

\section{Numerical Results}\label{secSimulation}

In this section, we present numerical results to evaluate the
proposed algorithms. For all simulations presented, we choose the
total number of randomization to be $L=200$.

{\subsection{Grouping for Admission Control and Scheduling}}

We evaluate the performance of the proposed algorithm for joint
admission control and beamforming. Let $N>0$ be a constant, which
represents the number of antennas at the receiver. We generate a
single cell network with radius $500$ meters, and with the
BS/receiver located in the cell center. The location for the
transmit nodes are randomly generated within the cell, and  are at
least $100$ meters away from the receiver. Let $d_i$ denote the
distance between node $i$ and the receiver and let
${\bh}_n\in\mathbb{C}^{M\times 1}$
 denote  the channel vector for the path between
the $i$-th node and the $n$-th receive antenna. We model the $i$-th
entry of ${\bh}_n$  as a zero mean circularly symmetric complex
Gaussian variable with variance (per real/imaginary dimension) given
by $\left({200}/{d_{i}}\right)^{3.5}L_{i}$, where
$10\log10(L_{i})\sim\mathcal{N}(0,64)$ is a real Gaussian random
variable modeling the shadowing effect. Set
$\bR=\sum_{n=1}^{N}\bh_n\bh^H_n$. Suppose the network has $M=50$
transmit nodes, with each node having the same transmit power
$P=-10$ dBW. The proposed algorithm in Section
\ref{secAdmissionControl} (abbreviated as Alg.1) is compared with a
sparse PCA (SPCA) based algorithm, whose main steps are listed
below: {\it i)} Approximately find a {\it sparse} principal
component of $\bR$ with $Q$ non-zero entries (denoted as $\bhw$),
using the backward-forward algorithm proposed in
\cite{moghaddam:spectral}; {\it ii)} normalize $\bhw$ by a constant
$\epsilon$ so that all the individual power constraints are
satisfied. In its first step, this algorithm tries to find the set
of nodes that, when replacing the individual power constraints with
a total power constraint, can provide the (approximately) best
averaged receive SNR.

Table \ref{tableGeneralChannel} demonstrates the maximum, the
minimum and the averaged ratios achieved for running the algorithms
over $500$ independent random generations of the channel. Note that
the approximation ratio for a solution $\bw^*$ is calculated by
$\alpha=\frac{v_1^{\rm CP}(\bw^*)}{v_1^{\rm SDP}(\bX^*)}$.

\begin{table*}[htb]
\vspace*{-0.1cm} \caption{Approximation Ratio of The Proposed and
SPCA Algorithms}\vspace*{-0.2cm}
\begin{center}
\small{
\begin{tabular}{|c |c| c|c|c| c|c|} \hline
 &{\bf  Alg.1 Min} & {\bf Alg.1 Mean} & {\bf Alg.1 Max}& {\bf SPCA Min}& {\bf SPCA Mean}& {\bf SPCA Max}\\
 \hline
 $N=5$, $Q=\lceil{M}/{3}\rceil$ & 1.12&  1.20 & 1.31 & 2.88 & 5.85 & 9.51\\
 \hline
 $N=10$, $Q=\lceil{M}/{3}\rceil$ & 1.11&  1.21 & 1.51& 2.48 & 6.74 & 10.62\\
\hline
 $N=15$, $Q=\lceil{M}/{3}\rceil$ & 1.11 &1.20 &1.43 & 2.71& 6.79 &11.27\\
 \hline
 \hline
 $N=5$, $Q={M}/{2}$ & 1.09& 1.14 & 1.49 & 3.04 & 6.97 & 12.13\\
 \hline
 $N=10$, $Q={M}/{2}$ & 1.08&  1.18 & 1.72& 2.81 & 7.94 & 14.08\\
\hline
 $N=15$, $Q={M}/{2}$ & 1.07 &1.17 &1.87 & 3.86& 8.28 &13.74\\
 \hline
\end{tabular} } \label{tableGeneralChannel}
\end{center}
\vspace*{-0.1cm}
\end{table*}

In Fig. \ref{figRatioM}--\ref{figObjM} we plot the performance of
the algorithms for different sizes of the network. For a given
network size, we choose $Q=10$ and let $N=5$. For each network
$(Q,M)$ pair, the algorithm is run for $500$ independent
realizations of the network. We again plot the maximum, the minimum
and the averaged approximation ratios achieved among those $500$
realizations. We see that the proposed algorithm achieves very low
worst-case approximation ratio, which suggests that high SNR
performance is obtained for almost all Monte Carlo runs. In Fig.
\ref{figObjQ}, a similar experiment is conducted with the number of
transmit nodes fixed at $M=50$, but with varying cooperative group
size $Q\in[6,~15]$.

   \begin{figure*}[ht]
    \begin{minipage}[t]{0.48\linewidth}\vspace*{-0.2cm}
    \centering
     {\includegraphics[width=
1\linewidth]{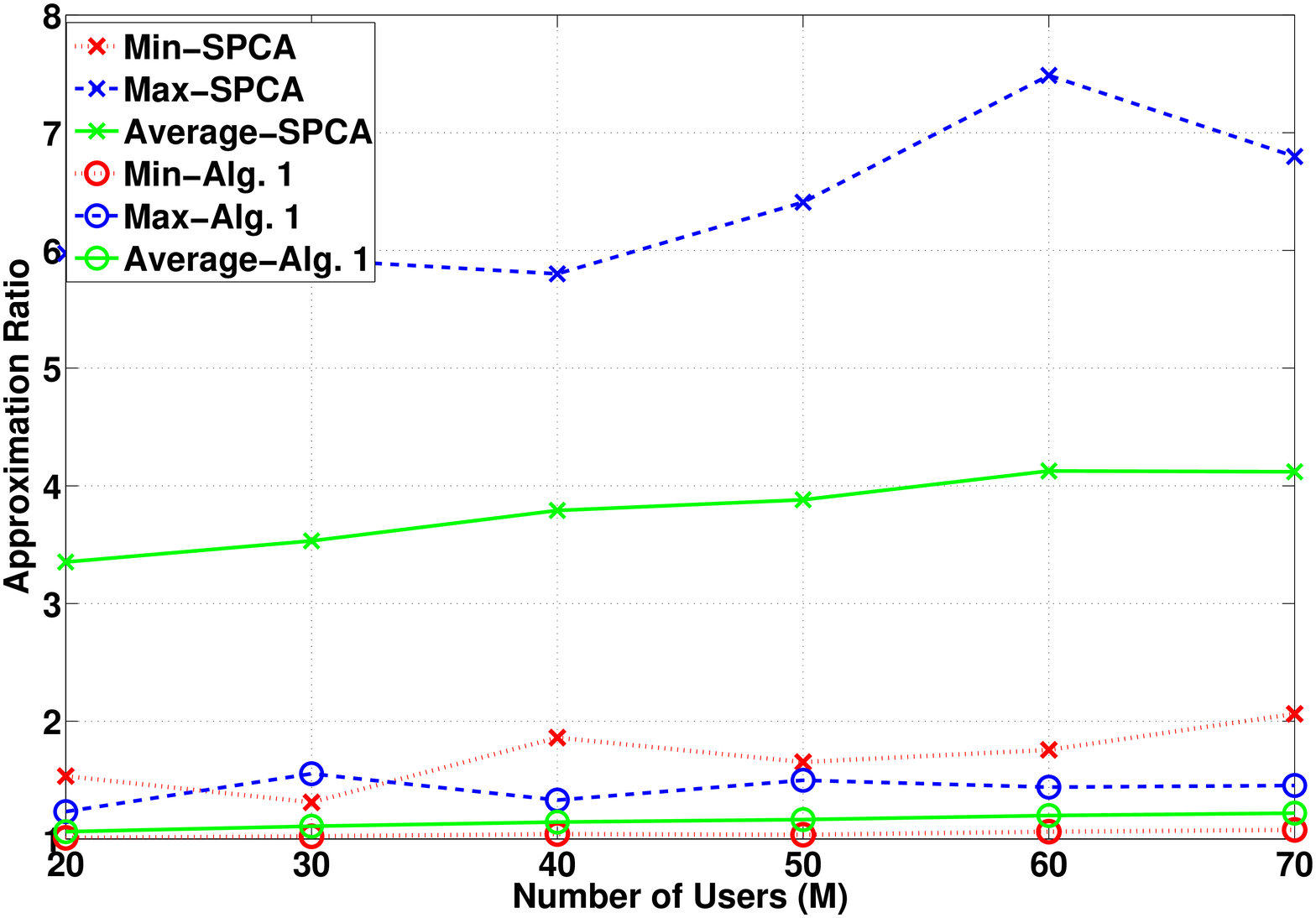}
\caption{\footnotesize Approximation ratio for admission control
with different network sizes. $M\in[10, 20, 30, 40, 50, 60, 70]$,
$Q=10$, $P=-10$dBW, $N=5$.
}\label{figRatioM} }\vspace*{-0.2cm}
\end{minipage}\hfill
    \begin{minipage}[t]{0.48\linewidth}\vspace*{-0.2cm}
    \centering
         {\includegraphics[width=
1\linewidth]{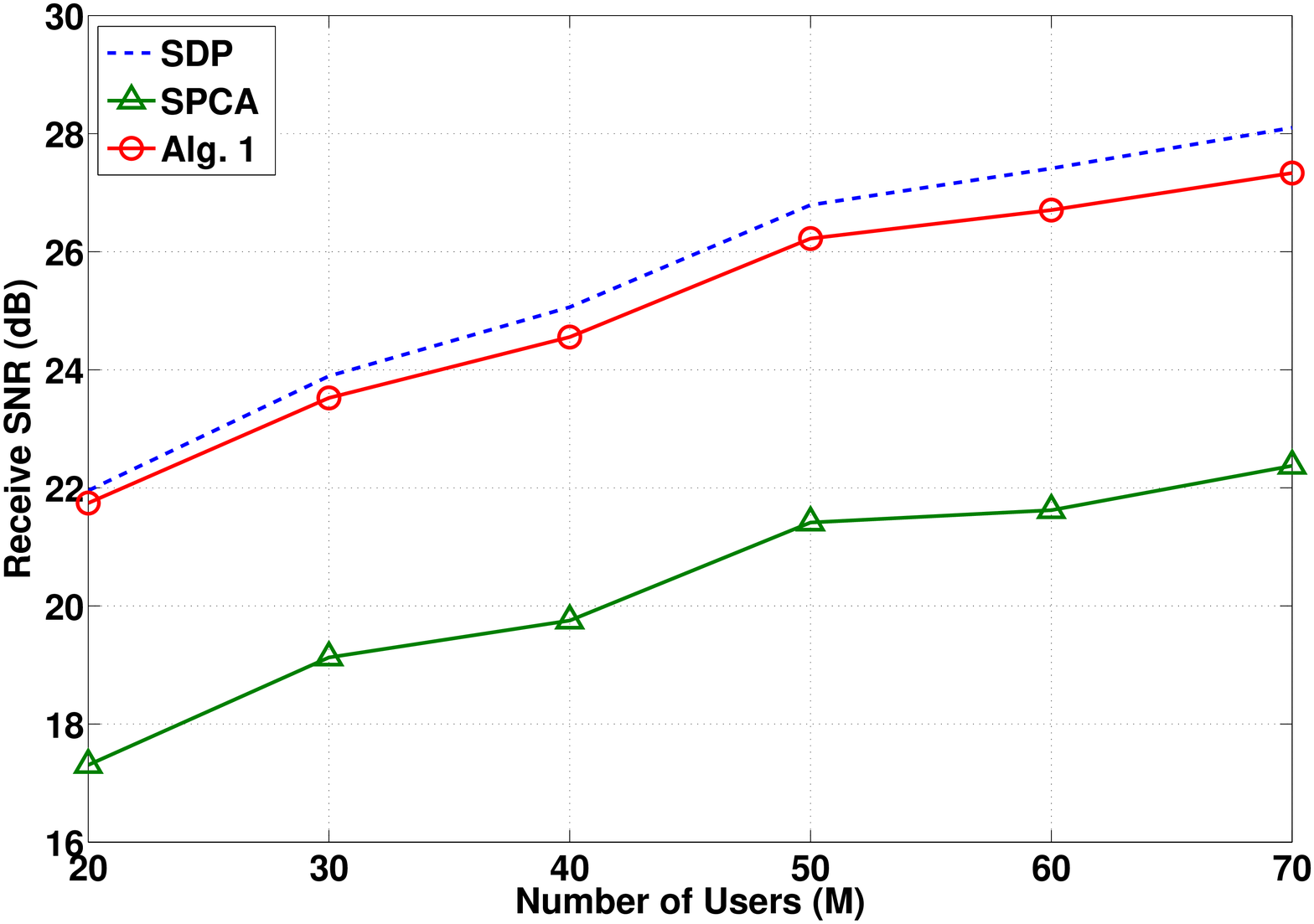} \caption{\footnotesize
Receive SNR for admission control with different network sizes.
$M\in[10, 20, 30, 40, 50, 60, 70]$, $Q=10$, $P=-10$dBW, $N=5$. 
}\label{figObjM} } \vspace*{-0.2cm}
\end{minipage}
    \end{figure*}

       \begin{figure*}[ht]
        \begin{minipage}[t]{0.48\linewidth}\vspace*{-0.1cm}
    \centering
    {\includegraphics[width=
1\linewidth]{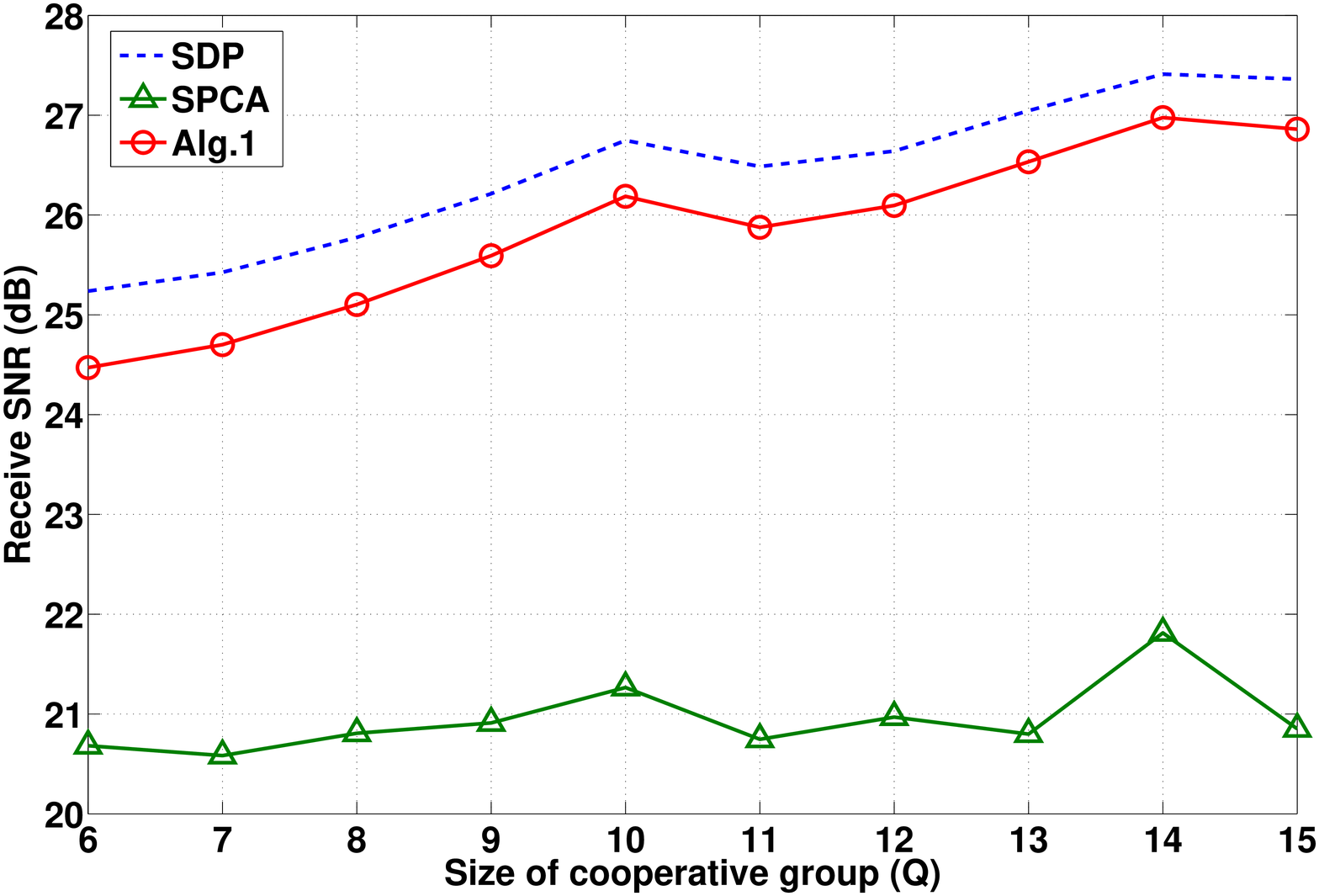} \caption{\footnotesize
Receive SNR for admission control with different cooperative group
sizes. $Q\in[6,~15]$, $M=50$, $P=-10$dBW, $N=5$.
}\label{figObjQ} }\vspace*{-0.2cm}
\end{minipage}
    \begin{minipage}[t]{0.48\linewidth}\vspace*{-0.1cm}
    \centering
        {\includegraphics[width=
1\linewidth]{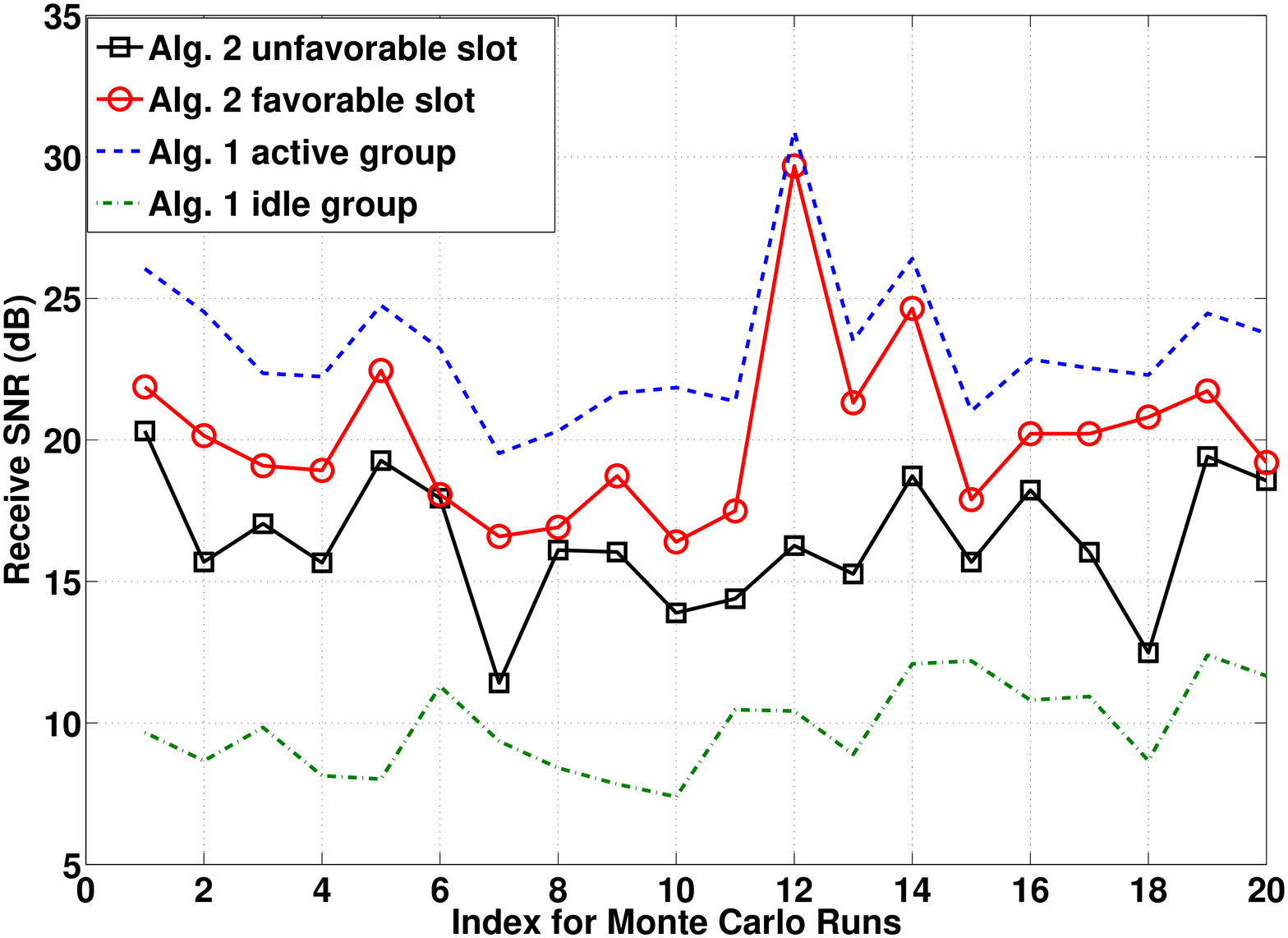}
\caption{\footnotesize Receive SNRs in different slots.
}\label{figFairness} } \vspace*{-0.2cm}
\end{minipage}\hfill
    \end{figure*}

%

In Table \ref{tableGeneralChannelScheduling}, we show the
performance of the proposed algorithm in a network of $M=30$
transmit nodes, for the max-min scheduling problem (abbreviated as
Alg.2). The proposed algorithm is compared with the following two
heuristic benchmarks: 1) randomly partition the nodes into two
groups of size $Q$ and $M-Q$, for each of which we solve a PCA
followed by a normalization step (abbreviated as R-PCA); 2) randomly
partition the nodes into two groups of size $Q$ and $M-Q$, and
obtain a solution $\bw_1^*$ and $\bw_2^*$ following steps S4)--S8)
in Table \ref{tableAlgorithm} (abbreviated as R-SDR).
Fig.~\ref{figFairness} illustrates the effectiveness of Alg.2 in
balancing the receive SNRs for different slots. We plot the receive
SNRs computed by Alg.2 in both slots (referred to the {\it
favorable}/{\it unfavorable} slot in the figure) during $20$ Monte
Carlo runs of the algorithm. For comparison, the plots are overlaid
with those obtained by running Alg.1 in the same network. For the
latter case, the receive SNRs are shown for both the active nodes
and the {\it idle nodes}. For the set of idle nodes that are
excluded from the cooperative group, their virtual beamformer is
computed using steps S4)--S8) in Table II, with $\cS$ replaced by
$\cbS$, and with all the matrices $\bU$, $\bE_i$ and $\bfDelta$
computed using $\cbS$. Fig.~\ref{figFairness} shows that for the
admission control formulation, when the idle nodes are offered a
chance of being served, they can only achieve very low SNR. On the
contrary, the scheduling formulation results in more balanced SNRs
for both slots.

In the aforementioned numerical results, the proposed SDR algorithms
can clearly deliver high quality approximate solutions (in terms of
both the averaged and the worse case performance) to the joint
admission control/scheduling and VB problem.

\begin{table*}[htb]
\vspace*{-0.1cm} \caption{Approximation Ratio of the Algorithm for
Scheduling}\vspace*{-0.5cm}
\begin{center}
\small{
\begin{tabular}{|c |c| c|c|c| c|c|} \hline
 &{\bf Alg. 2 Mean} & {\bf Alg. 2 Max}& {\bf R-PCA Mean}& {\bf R-PCA Max} & {\bf R-SDR Mean}& {\bf R-SDR Max} \\
 \hline
 N=5, $Q=\lceil{M}/{4}\rceil$  & 1.88 &4.06 &9.76& 24.84& 3.23 &10.14\\
 \hline
 N=10, $Q=\lceil{M}/{4}\rceil$ & 2.77 &9.55 &10.87& 28.56& 5.68 &15.74\\
 \hline
\end{tabular} } \label{tableGeneralChannelScheduling}
\end{center}
\vspace*{-0.2cm}
\end{table*}

\vspace{-0.3cm}
\subsection{Grouping for Relay Selection in Relay Networks}
In this section, we numerically evaluate the performance of the SDR
algorithm for solving the joint relay selection and VB problem
(abbreviated as Alg.3).

For the relay network, the channel covariances are generated as
follows \cite{Nassab08, Dehkordy09}. Assume $f_i$ can be written by
$f_i=\hat{f}_i+\tilde{f}_i$, where $\hat{f}_i$ is the mean of $f_i$
and $\tilde{f}_i$ is a zero-mean random variable. Assume that
$\tilde{f}_i$, $\tilde{f}_j$ are independent for $i\ne j$, and
choose $\hat{f}_i=e^{j\theta_i}/\sqrt{\eta_{f_i}}$ and ${\rm
var}(\tilde{f})=\eta_{f_i}/(1+\eta_{f_i})$, where $\theta_i$ is a
uniform random variable on the interval $[0,\ 2\pi]$, and
$\eta_{f_i}$ is a parameter that determines the level of uncertainty
in the channel coefficient $f_i$. Similarly, let
$g_i=\hat{g}_i+\tilde{g}_i$, with $\hat{g}_i$ and $\tilde{g}_i$
defined similarly as $\hat{f}_i$ and $\tilde{f}_i$. 
The justification of this channel model can be found in
\cite{Nassab08,Dehkordy09}. In our simulations, $\{\eta_{f_i}\}$ and
$\{\eta_{g_i}\}$ are generated randomly from $-10$ dB to $10$ dB,
accounting for different uncertainty levels for different nodes.

   \begin{figure*}[ht]
    \begin{minipage}[t]{0.48\linewidth}\vspace*{-0.2cm}
    \centering
     {\includegraphics[width=
1\linewidth]{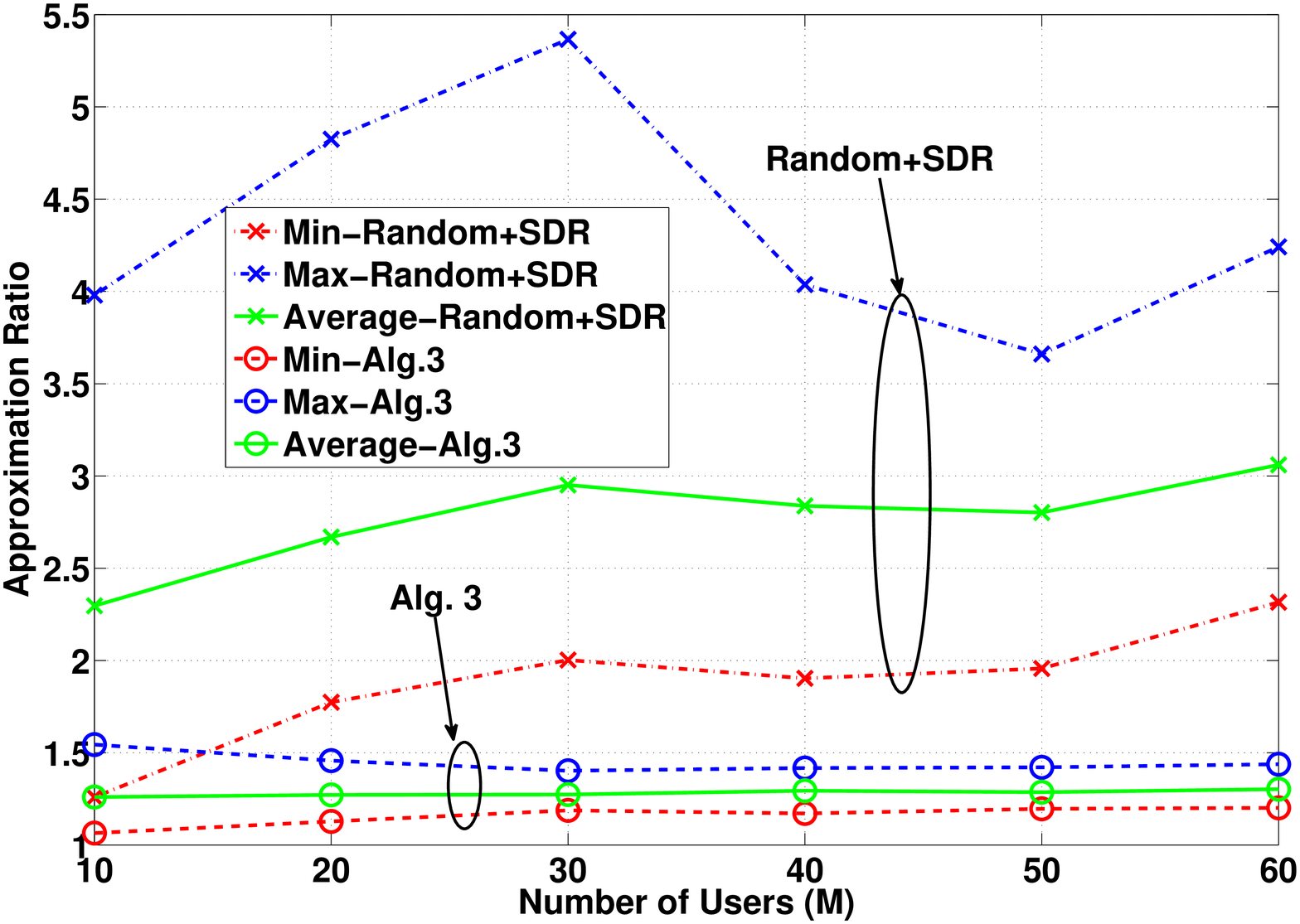}
\caption{\footnotesize Approximation ratio with different network
sizes. $M\in[10, 20, 30, 40, 50, 60]$, $Q=(\frac{M}{2})$, $P_{\rm
tot}=10{\rm dBW}$, $100$ independent realizations of
channel.}\label{figfigRatioRelayChangeMRelay1} \vspace*{-0.2cm}}
\end{minipage}\hfill
    \begin{minipage}[t]{0.48\linewidth}\vspace*{-0.2cm}
    \centering
    {\includegraphics[width=
1\linewidth]{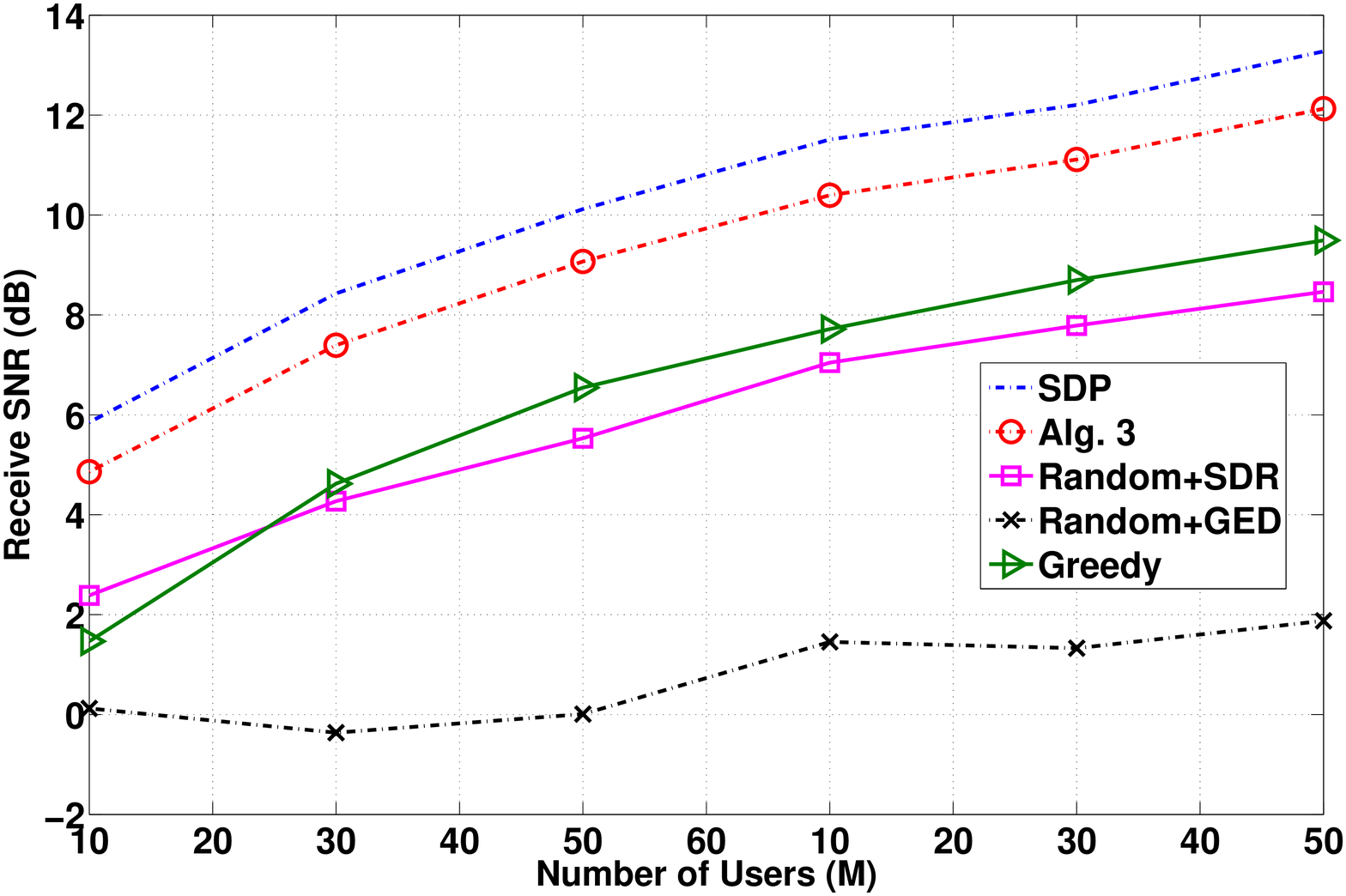}
\caption{\footnotesize Averaged received SNR for different
algorithms. $M\in[10, 20, 30, 40, 50, 60]$, $Q=(\frac{M}{2})$,
$P_{\rm tot}=10{\rm dBW}$, $100$ independent realizations of
channel.}\label{figSNRRelayChangeM}\vspace*{-0.2cm} }
\end{minipage}
    \end{figure*}

Our SDR algorithm is compared against the following three
algorithms:

\begin{itemize}
\item  {\it Random-GED}: This is a variant of the generalized eigenvalue decomposition (GED) based
algorithm proposed in \cite[Section IV-B]{Nassab08}. In its original
form, all the relays are utilized, and the algorithm computes an
approximate solution of the max-SNR problem by performing a PCA for
the matrix $\bS^{-1}\bF$, followed by a normalization step to ensure
the individual power constraints. To incorporate the selection of
relays, we first {\it randomly} select $Q$ out of $M$ relays, and
then perform the PCA and normalization steps.

\item {\it Random-SDR}: In this algorithm, we first {\it randomly} select $Q$ out of
$M$ relays in the network, and then perform the two-step SDR
algorithm with the fixed node grouping (see the algorithm
description in Section \ref{secRelay}).

\item {\it Greedy algorithm}: This algorithm largely
follows from the one proposed in \cite{Jing09}, in which nodes are
added to the cooperation set successively and greedily as long as it
can improve the received SNR level, or  the required group size is
less than $Q$.
\end{itemize}

In Fig.
\ref{figfigRatioRelayChangeMRelay1}--\ref{figSNRRelayChangeM}, the
approximation ratio and the achieved receive SNR of different
algorithms are plotted against the size of the network. Fig.\
\ref{figRatioRelayChangeSNR}--\ref{figSNRRelayChangeSNR} show the
performance of the algorithms when we increase the relays'
transmission powers. In these two sets of figures we set $P_{0}=0$
dBW, $Q=\frac{M}{2}$, and set the total transmit power of the relays
$P_{\rm tot}=10$ dBW. The curves labeled by ``SDP" in
Fig.~\ref{figSNRRelayChangeM} and Fig.~\ref{figSNRRelayChangeSNR}
represent the upperbounds of the achievable receive SNR which are
obtained by solving the relaxed problem \eqref{problemSDPRelay}. We
see that the proposed SDR algorithm is close to the optimal solution
across all power levels and   all network sizes, while the other
three algorithms perform notably worse in these experiments.
Furthermore, the proposed algorithm has significantly better ``worst
case" performance than the Random-SDR algorithm (see
Fig.~\ref{figfigRatioRelayChangeMRelay1} and
Fig.~\ref{figRatioRelayChangeSNR}).

   \begin{figure*}[ht]
    \begin{minipage}[t]{0.48\linewidth}\vspace*{-0.2cm}
    \centering
     {\includegraphics[width=
1\linewidth]{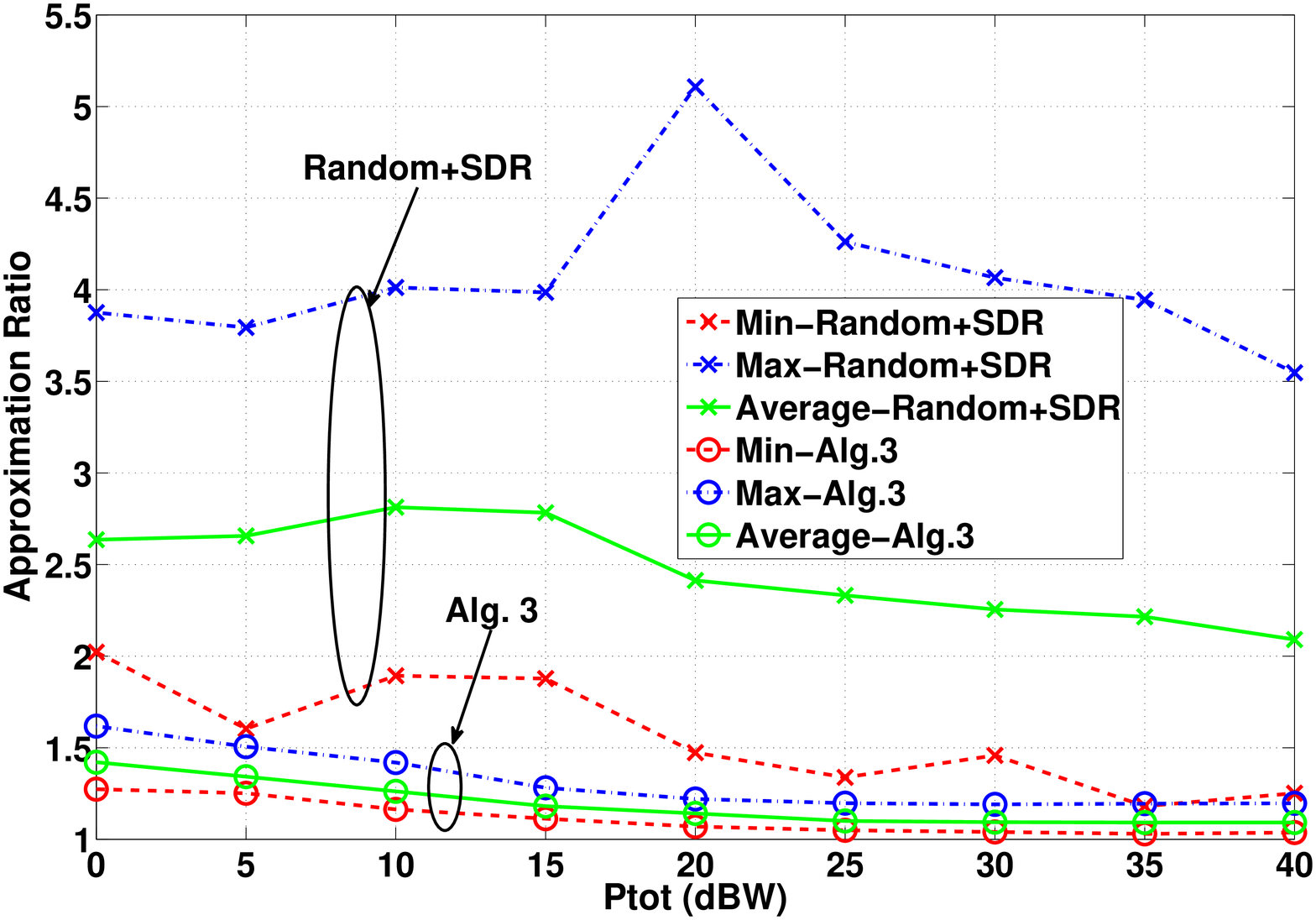}
\caption{\footnotesize Approximation ratio for different algorithms.
$M=30$, $Q=(\frac{M}{2})$, $P_{\rm tot}=[0, 5, 10, 15, 20, 25, 30,
35, 40]$ dBW, $100$ independent realizations of
channel.}\label{figRatioRelayChangeSNR} \vspace*{-0.2cm}}
\end{minipage}\hfill
    \begin{minipage}[t]{0.48\linewidth}\vspace*{-0.2cm}
    \centering
    {\includegraphics[width=
1\linewidth]{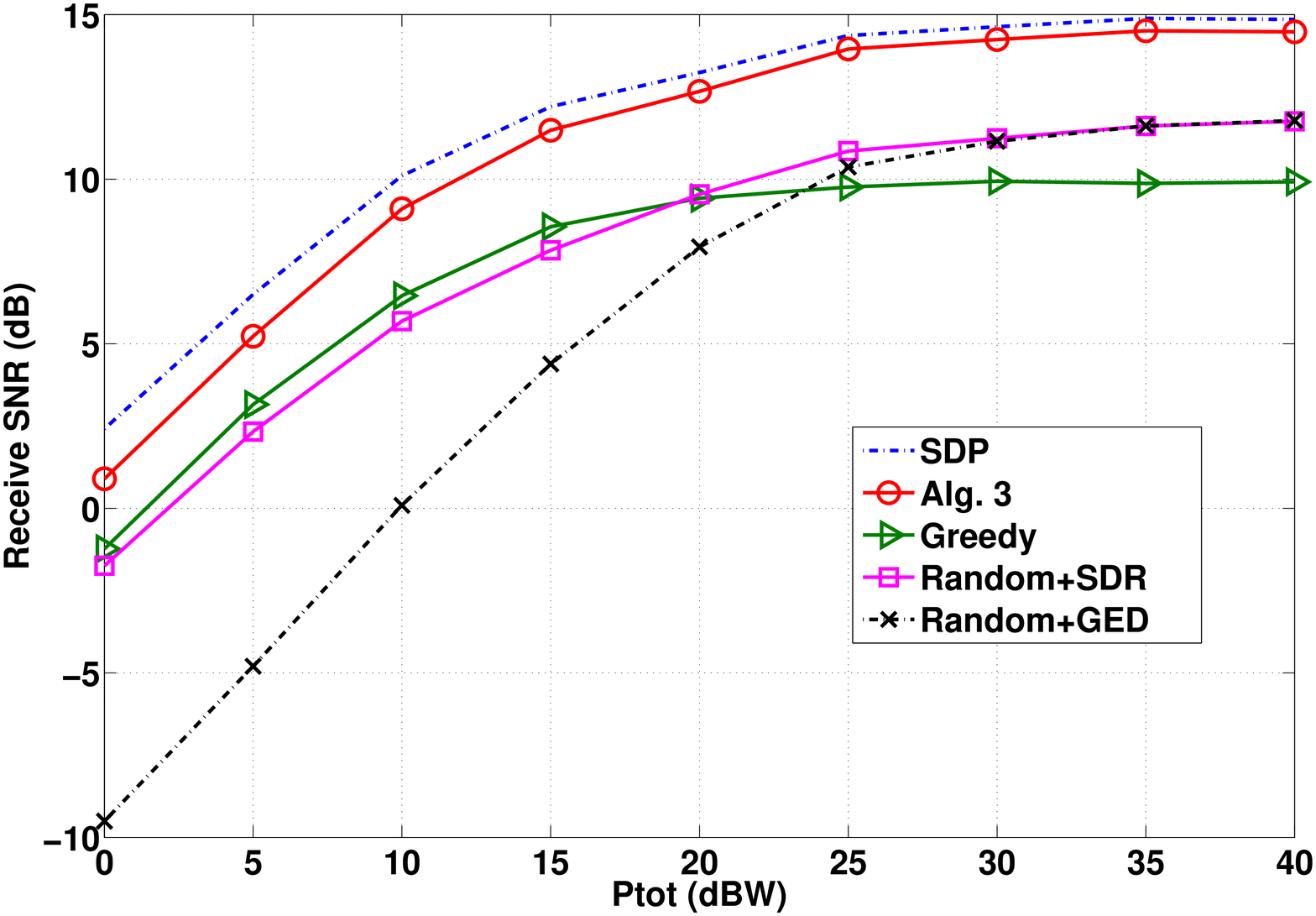}
\caption{\footnotesize Averaged receive SNR for different
algorithms. $M=30$, $Q=(\frac{M}{2})$, $P_{\rm tot}=[0, 5, 10, 15,
20, 25, 30, 35, 40]$ dBW, $100$ independent realizations of
channel.}\label{figSNRRelayChangeSNR} \vspace*{-0.2cm}}
\end{minipage}
    \end{figure*}

\section{Final Remarks}\label{secConclusion}

This paper addresses the joint node grouping and the virtual
beamformer design problem  under
different network setup. 
We formulate this problem as a cardinality constrained quadratic
optimization problem. Unfortunately,  the resulting optimization
problem is computationally intractable as shown in
Table~\ref{tableSummary} below.

\begin{table*}[htb]
\vspace*{-0.1cm} \caption{Summary of Complexity
Results}\vspace*{-0.3cm}
\begin{center}
\small{
\begin{tabular}{|c |c| c|c|c| c|c|} \hline
&\multicolumn{3}{|c|} {\bf Admission Control} & \multicolumn{3}{|c|}
{\bf Scheduling
(2-Slot)}\\
 \hline
 {\bf Channel} $\bR$ &{\bf  General} & {\bf Diagonal } & {\bf Rank 1}& {\bf General }& {\bf Diagonal}& {\bf Rank 1}\\
 \hline
 \hline
{\bf VMIMO} & NP-Hard & Ploy. Time Solvable & Ploy. Time Solvable &  NP-Hard & NP-Hard & NP-Hard \\
\hline
{\bf VMIMO-Relay} & NP-Hard & Ploy. Time Solvable & Unknown &  NP-Hard & NP-Hard & Unknown\\
\hline
\end{tabular} } \label{tableSummary}
\end{center}
\vspace*{-0.5cm}
\end{table*}

Given the NP-hardness of the problem, we develop several novel
approximation algorithms based on the semidefinite programming
relaxation (SDR) technique. The quality of the approximate solutions
is evaluated both theoretically and via numerical simulation.
Somewhat surprisingly, we are able to prove that the resulting SDR
algorithm for admission control has a guaranteed approximation
performance in terms of the gap to global optimality, regardless of
channel realizations. Such theoretical analysis lends strong support
to the practical performance of the proposed algorithms. Indeed,
their effectiveness has been confirmed in our extensive numerical
experiments.

In closing, we suggest several directions of future research. First,
it will be interesting to extend the theoretical analysis for the
approximation bounds to the scheduling and the relay selection
problems. Second, we expect that our SDR approach can be applied and
extended to some new application scenarios. For example, in the
admission control problem, to mitigate the interference to a
neighboring co-existing system, one could further impose a {\it
total interference} constraint of the form
$\mathbb{E}[|\bw^H\mathbf{g}|^2]\le I$ to the selected group of
nodes, where $\mathbf{g}$ represents the channel between the
transmit nodes and the BS of the neighboring system. 
Additionally, the proposed SDR approach may prove to be useful in
designing future multi-cell cellular systems, where not all, but a
subset of BSs can jointly transmit and receive signals for a given
mobile user (see the recent work \cite{hong12sparse, kim12,
marsch11clustering} on the partial CoMP technique).

\appendix
\subsection{Proof of  Claims~\ref{remarkObjective} and
\ref{remarkFeasibility}}\label{appendixRemarkFeasibility} { Claim
\ref{remarkObjective} is due to the following chain of
equalities{\small
\begin{align}
v_1^{\rm CP}(\bw^{(\ell)})&=(\bw^{(\ell)})^H\bR\bw^{(\ell)}
=\frac{1}{(t^{(\ell)})^2}(\bfxi^{(\ell)})^T\bU^H\bfDelta\bR[\cS]\bfDelta^H\bU\bfxi^{(\ell)}\nonumber\\
&\stackrel{(a)}=\frac{1}{(t^{(\ell)})^2}(\bfxi^{(\ell)})^T\bfSigma\bfxi^{(\ell)}
\stackrel{(b)}=\frac{1}{(t^{(\ell)})^2}\trace[\bfSigma]\stackrel{(c)}
=\frac{1}{(t^{(\ell)})^2}\trace[\bR[\cS]\bY^*]
\end{align}}
\hspace{-0.2cm}where $(a)$ and $(c)$ are from \eqref{eqSigma}; $(b)$
is from the fact that $\bfSigma$ is diagonal, and the diagonal
elements of $\bfxi^{(\ell)}(\bfxi^{(\ell)})^T$ are all $1$.}

To argue the feasibility condition claimed in Claim
\ref{remarkFeasibility}, we first show that the constraint
$(\bx^{(\ell)})^H\bD_i\bx^{(\ell)}\le 1$ is always satisfied for all
$i\in\cM$. To see this, we consider the following two cases (we omit
the superscript $(\ell)$ for notational simplicity).

{\bf Case i):} If $i\in\cS$, then we have{\small
\begin{align}
\bx^H\bD_i\bx&=(\bw)^H\bC_{i,1}\bw+(\bx_0)^H\bC_{i,0}\bx_0\nonumber\\
&=(\bw[\cS])^H\bC_{i,1}[\cS]\bw[\cS]+(\bx_0)^H\bC_{i,0}\bx_0\nonumber\\
&=\frac{1}{t^2}\bfxi^T\bU^H\bfDelta\bC_{i,1}[\cS]\bfDelta^H\bU\bfxi
+\frac{1}{2}(1-\bx_0[i]\bx_0[M+1])\nonumber\\
&\stackrel{(a)}=\frac{1}{t^2}\bfxi^T\bU^H\bE_i\bU\bfxi\stackrel{(b)}\le
1\nonumber
\end{align}}
\hspace{-0.2cm}where $\rm (a)$ is from the definition of $\bE_i$ and
from the fact that $\bx_0[M+1]=1$ and for all $i\in\cS$,
$\bx_0[i]=1$; $(b)$ is from the definition of $t$ in Step S6) of the
algorithm.

{\bf Case ii):} If $i\in\bar{\cS}$, then from Step S7) of the
algorithm, we have $w_i=0$. Using the fact that
$\bC_{i,1}=\be_i\be^T_i\frac{1}{P}$, we have: {\small $
\bx^H\bD_i\bx=(\bw)^H\bC_{i,1}\bw=|w_i|^2\frac{1}{P}=0.$} It is
straightforward to see that the cardinality constraint $\bx^H\bB\bx=
4Q$ is satisfied for each solution $\bw^{(\ell)}$, as this
constraint mandates that $M-Q$ nodes do not transmit.

In summary, we conclude that $\bx$ is feasible for the problem
\eqref{problemReformulate}. Due to the equivalence between problems
\eqref{problemReformulate} and \eqref{problemCCQP}, $\bw$ is also
feasible for the latter problem.

\subsection{Proof of Theorem
\ref{thmAppRatio}}\label{appProofTheorem1}
\begin{proof}
Observe that it is sufficient to show that with high probability
{\small $ v_1^{\rm CP}(\bw^{(\ell)})\ge \frac{1}{\alpha_1}v_1^{\rm
SDP}$}. Below we first show that for any $\ell=1,\cdots, L$, there
exists a finite constant $\alpha_1>1$ satisfying: {\small $
\Prob\left(v_1^{\rm CP}(\bw^{(\ell)})\ge \frac{1}{\alpha_1}v_1^{\rm
SDP}\right)\ge\delta>0 $}, or equivalently,{\small
\begin{align}
\Prob\left(\frac{1}{(t^{(\ell)})^2}\trace[\bR[\cS]\bY^{*}]\ge
\frac{1}{\alpha_1}\trace[\bR\bX^{*}_1]\right)\ge\delta>0.\label{eqDesiredBound}
\end{align}}
\hspace{-0.2cm}That is, with positive probability, the solution
$\bw^{(\ell)}$ generated by the proposed algorithm is at least as
good as $\frac{1}{\alpha_1}$ fraction of $v^{\rm SDP}$. If this is
indeed true, it follows that the probability that the solution
$\bw^*$ achieves an objective that is at least $\frac{1}{\alpha_1}$
of $v_1^{\rm SDP}$ is given by{\small
\begin{align}
&\Prob\left(\max_{\ell}v_1^{\rm CP}(\bw^{(\ell)})\ge
\frac{1}{\alpha_1}v_1^{\rm
SDP}\right)\nonumber\\
&=1-\Prob\left(\max_{\ell}v_1^{\rm CP}(\bw^{(\ell)})<
\frac{1}{\alpha_1}v_1^{\rm SDP}\right)\nonumber\\
&=1-\prod_{l=1}^{L}\Prob\left(v_1^{\rm CP}(\bw^{(\ell)})<
\frac{1}{\alpha_1}v_1^{\rm SDP}\right)\ge 1-(1-\delta)^L\nonumber.
\end{align}}
\hspace{-0.2cm}Clearly, this probability approaches $1$
exponentially as $L$ becomes large.

Below we will show \eqref{eqDesiredBound}. The analysis is divided
into the following three steps.

{\bf Step 1)}: Let $\beta>0$ be a constant. We first show that when
$\frac{\beta}{\alpha_1}$ is small enough, with zero probability the
event $\trace[\bR[\cS]\bY^{*}]\le
\frac{\beta}{\alpha_1}\trace[\bR\bX_1^{*}]$ happens.

To this end, we first lower bound $\trace[\bR[\cS]\bY^{*}]$. We have
the following two cases.

{\bf Case i)} Suppose $\cS=\cT=\{j:\bX^{*}_0[j,M+1]\ge
\underline{\bx}_Q\}$, then we have{\small
\begin{align}
v_2^{\rm SDP}=\trace[\bR[\cS]\bY^{*}]\ge
\trace[\bR[\cS]\bX^{*}_1[\cS]]\nonumber 
\end{align}}
\hspace{-0.2cm}where the inequality is from that fact
$\bX^{*}_1[\cS]$ is a feasible solution to the problem
\eqref{problemSDP-QCQP}, and that $\bY^*$ is the optimal solution
for that problem. From Step S2) of the algorithm, we must have
{\small\begin{align} \trace[\bR[\cS]\bY^{*}]\ge
\frac{QP}{M}\trace[\bR].
\end{align}}
{\bf Case ii)} Suppose $\cS=\{j:\bR[i,i]\ge \underline{\br}_Q\}$.
Then we have{\small
\begin{align}
\trace\left[\bR[\cS]\bY^{*}\right]\stackrel{\rm (i)}\ge
P\trace[\bR[\cS]\bI_Q]\stackrel{\rm (ii)}\ge
\frac{QP}{M}\trace[\bR]\label{eqLowerBoundHX}
\end{align}}
\hspace{-0.2cm}where ${\rm (i)}$ is again from that fact $P\bI_{Q}$
is a feasible solution to the problem \eqref{problemSDP-QCQP}, and
that $\bY^*$ is the optimal solution for that problem; ${\rm (ii)}$
is from the fact that each $i\in\cS$ is among the largest $Q$
elements in the set $\{\bR[i,i]\}_{i=1}^{Q}$, which leads to
$\sum_{i\in\cS}\bR[i,i]\ge\frac{Q}{M}\sum_{i=1}^{M}\bR[i,i]=\frac{Q}{M}\trace[\bR]$.


We then upper bound $\trace[\bR\bX^{*}_1]$. By a trace inequality
for the product of two semi-definite matrices, we have that
\cite{Marshall10} {\small
\begin{align}
\trace[\bR\bX^{*}_1]&\le\sum_{i=1}^{M}\lambda_{i}\left(\bX^{*}_1\right)\lambda_i(\bR)
\le
\lambda_1(\bR)\sum_{i=1}^{M}\lambda_{i}\left(\bX^{*}_1\right)=\lambda_1(\bR)\trace[\bX^{*}_1]\label{eqUpperBoundTrace}.
\end{align}}
\hspace{-0.2cm}Utilizing this result, we have {\small
\begin{align}
\trace[\bR\bX_1^{*}]&\le
\lambda_{1}(\bR)\trace[\bX_1^{*}]\stackrel{\rm
(i)}=P\lambda_{1}(\bR)
\sum_{i=1}^{M}\left(\frac{1}{2}+\frac{1}{2}\bX^{*}_0[i,M+1]\right)\nonumber\\
&=P\lambda_{1}(\bR)\left(\frac{M}{2}+\frac{1}{2}\sum_{i=1}^{M}\bX^{*}_0[i,M+1]\right)\nonumber\\
&\stackrel{\rm
(ii)}=P\lambda_{1}(\bR)\left(\frac{2Q-M}{2}+\frac{M}{2}\right)=QP\lambda_{1}(\bR)\label{eqUpperBoundHX}
\end{align}}
\hspace{-0.2cm}where ${\rm (i)}$ is from the tightness of the first
set of constraints of the problem \eqref{problemSDP} (cf.
\eqref{eqPowerConstraintTight}) and ${\rm (ii)}$ is due to Claim
\ref{remarkTight} and Claim \ref{remarkXLowerBound}. Comparing
\eqref{eqLowerBoundHX} and \eqref{eqUpperBoundHX}, we see that
choosing
$\frac{\beta}{\alpha_1}\le\frac{\trace[\bR]}{M\lambda_{1}(\bR)}$
ensures{\small
\[\Prob\left(\trace[\bR[\cS]\bY^*]\le
\frac{\beta}{\alpha_1}\trace[\bR\bX^{*}_1]\right)=0.\]} {\bf Step
2)}: For fixed $\alpha_1$, we bound the probability that
$\Prob\left(\frac{1}{(t^{(\ell)})^2}\le \frac{1}{\beta}\right)$ as
follows (omitting $(\ell)$ for simplicity, and defining
$\bhE_i\triangleq\bU^T\bE_i\bU$){\small
\begin{align}
\Prob\left(\frac{1}{t^2}\le \frac{1}{\beta}\right)&=\Prob(t^2\ge
\beta)=\Prob\left(\max_{i\in\cS}(\bfxi)^T\bhE_i\bfxi\ge\beta\right)\nonumber\\
&\le \sum_{i\in\cS}\Prob\left((\bfxi)^T\bhE_i\bfxi\ge\beta\right)<
4Q \mu\exp\left(-\frac{\beta}{8}\right)\nonumber
\end{align}}
where $\mu=\min[Q, \max_{i}\rank(\bhE_i)]$. The last inequality is
obtained by slightly generalizing the existing result in
\cite[Proposition 1]{Nemirovski_Roos_Terlaky_1999} \footnote{The
cited result can be generalized from the real case to the complex
case using the complex form of the large deviation results from
Bernstein, see e.g., \cite[Theorem 4]{greenLargeDeviation}, and the
recent results by Zhang and So \cite{zhang:optimal11}.}. To
explicitly compute the value for $\mu$, note that by definition, we
have: {\small $ \bhE_i=\bU^H\bE_i\bU
=\bU^H\bfDelta\bC_{i,1}[\cS]\bfDelta^H\bU\nonumber$}. As a result,
$\rank(\bhE_i)\le 1$ as by definition $\rank(\bC_{i,1})=1$, and any
one of its principal submatrices must have rank at most $1$. We
conclude that $0\le\mu\le1$.

{\bf Step 3)}: Utilizing the above result, choose{\small
\[
\beta=8\ln(5Q), \
\alpha_1=\frac{8M\lambda_{1}(\bR)}{\trace[\bR]}\ln(5Q),\]} we can
bound the left hand side of \eqref{eqDesiredBound} as follows{\small
\begin{align}
&\Prob\left(\frac{1}{t^2}\trace[\bR[\cS]\bY^{*}]\ge
\frac{1}{\alpha_1}\trace[\bR\bX^{*}_1]\right)\nonumber\\
&\ge\Prob\left(\trace[\bR[\cS]\bY^{*}]\ge
\frac{\beta}{\alpha_1}\trace[\bR\bX^{*}_1], \frac{1}{t^2}\ge
\frac{1}{\beta}\right)\nonumber\\
&\ge 1-\Prob\left(\trace[\bR[\cS]\bY^{*}]\le
\frac{\beta}{\alpha_1}\trace[\bR\bX^{*}_1]\right)-\Prob\left(\frac{1}{t^2}\le
\frac{1}{\beta}\right)\nonumber\\
&=1-\Prob\left(\frac{1}{t^2}\le
\frac{1}{\beta}\right)>1-4Q\exp\left(-\ln(5Q)\right)=1-4Q/(5Q)=\frac{1}{5}.\nonumber
\end{align}}
In conclusion, the final approximation ratio is given by{\small
\begin{align}
\alpha_1=\frac{8M\lambda_{1}(\bR)}{\trace[\bR]}\ln(5Q)
=\frac{8M\lambda_{1}(\bR)}{\sum_{i=1}^{M}\lambda_{i}(\bR)}\ln(5Q)\le
8M\ln(5Q).\nonumber
\end{align}}
\hspace{-0.2cm}This completes the proof.
\end{proof}

\subsection{Proof of Proposition
\ref{thmNPHardRelay}}\label{appNPhardRelay}

\begin{proof}
It suffices to show that solving \eqref{problemRelay} is NP-hard
even when the active relays are known. In other words, it is
sufficient to show that solving the problem
\begin{align}
\max_{\bw}&\quad\frac{\bw^H\bS\bw}{\sigma^2_n+\bw^H\bF\bw}\label{temp1}\\
{\rm s. t.}&\quad |w_i|^2\left(P_0
\mathbb{E}[|f_i|^2]+\sigma_{\nu}^2\right)\le  P, \ i=1,\cdots, M,
\nonumber
\end{align}
is NP-hard. We prove by using a polynomial time reduction from the
integer equal partitioning problem. To this end, let us consider the
following system parameters:
\[
P_0 = 1, \quad P = 2,\quad \sigma_v^2 = 1,\quad f_i = 1, \forall i.
\]
Then the problem \eqref{temp1} can equivalently be written as{\small
\begin{align}
\max_{\bw}&\quad\frac{\bw^H \mathbb{E}[\bg\bg^H]\bw}{\sigma^2_n+\bw^H {\rm diag}\left(\mathbb{E}[\bg\bg^H]\right)\bw}\label{temp2}\\
{\rm s. t.}&\quad |w_i|^2\le 1, \ i=1,\cdots, M. \nonumber
\end{align}}
{\hspace{-0.2cm}}Let $t$ denote the objective value of the above
problem. In order to check the achievability of a particular value
$t$, we need to check the feasibility of the following set of
inequalities:{\small
\begin{equation}
\label{temp3}
\begin{split}
&\bw^H \left(\mathbb{E}[\bg\bg^H]- t {\rm diag}\left(\mathbb{E}[\bg\bg^H]\right) \right)\bw \geq t \sigma_n^2 \\
&|w_i|^2 \leq 1, \quad \forall i.
\end{split}
\end{equation}}
\hspace{-0.2cm}Therefore, it suffices to show the NP-hardness of
checking the achievability of \eqref{temp3}. To this end, we first
claim that for given a positive definite matrix $\mathbf{A}$, the
following set is non-empty
\[
\mathcal{T} = \left\{(t,\mathbf{X})\mid \mathbf{X} - t {\rm
diag}(\mathbf{X}) = \mathbf{A}, \mathbf{A} \succ 0, t >0, t \in
\mathbb{R}, \mathbf{X}\in \mathbb{R}^{M\times M}\right\}
\]
To justify this claim, consider the mapping $\phi_{\mathbf{A}}(t):
\mathbb{R}^+ \mapsto \mathbb{R}^{M\times M}$, where
$\phi_\mathbf{A}(t) = \mathbf{B}$ with
\[
\mathbf{B}_{ij} = \left\{\begin{array}{cc}
\mathbf{A}_{ij} & i \neq j\\
\frac{\mathbf{A}_{ii}}{1-t} & i = j\\
\end{array}
\right..
\]
Note that $\phi_{\mathbf{A}}(t)$ is continuous over $[0,1)$ and
$\phi_{\mathbf{A}}(0) = \mathbf{A}$. Therefore, for small enough
$t'>0$, we have $\phi_{\mathbf{A}}(t') = \mathbf{B}'\succ 0$. In
other words, $(t',\mathbf{B}') \in \mathcal{T}$. The above claim
implies that there exists a positive definite matrix $\mathbf{R} =
\mathbb{E} [\bg\bg^H]$  and a positive scalar $t$ such that
\[
\mathbb{E} [\bg\bg^H] - t {\rm diag} \left(\mathbb{E}
[\bg\bg^H]\right) = 2C\bI - \mathbf{c}\mathbf{c}^T,
\]
with $C = \|\mathbf{c}\|^2$. Therefore, using a similar argument to
the one in the proof of Proposition~\ref{propositionNPHardReal}
completes the proof.
\end{proof}

 \vspace{-0.0cm}
\bibliographystyle{IEEEbib}
{\footnotesize
\bibliography{ref}

\begin{thebibliography}{10}

\bibitem{Foschini06}
G.J. Foschini, K.~Karakayali, and R.A. Valenzuela,
\newblock ``Coordinating multiple antenna cellular networks to achieve enormous
  spectral efficiency,''
\newblock {\em IEE Proceedings Communications}, vol. 153, no. 4, pp. 548 --
  555, 2006.

\bibitem{gesbert10}
D.~Gesbert, S.~Hanly, H.~Huang, S.~Shamai, O.~Simeone, and W.~Yu,
\newblock ``Multi-cell {MIMO} cooperative networks: A new look at
  interference,''
\newblock {\em IEEE Journal on Selected Areas in Communications}, vol. 28, no.
  9, pp. 1380 --1408, 2010.

\bibitem{Irmer11}
R.~Irmer, H.~Droste, P.~Marsch, M.~Grieger, G.~Fettweis, S.~Brueck, H.-P.
  Mayer, L.~Thiele, and V.~Jungnickel,
\newblock ``Coordinated multipoint: Concepts, performance, and field trial
  results,''
\newblock {\em IEEE Communications Magazine}, , no. 2, pp. 102--111, 2011.

\bibitem{Caire03}
G.~Caire and S.~Shamai,
\newblock ``On the achievable throughput of a multiantenna {Gaussian} broadcast
  channel,''
\newblock {\em IEEE Transactions on Information Theory}, vol. 49, no. 7, pp.
  1691 -- 1706, 2003.

\bibitem{yu07perantenna}
W.~Yu and T.~Lan,
\newblock ``Transmitter optimization for the multi-antenna downlink with
  per-antenna power constraints,''
\newblock {\em IEEE Transactions on Signal Processing}, vol. 55, no. 6, pp.
  2646 --2660, 2007.

\bibitem{Spencer04}
Q.~H. Spencer, A.~L. Swindlehurst, and M.~Haardt,
\newblock ``Zero-forcing methods for downlink spatial multiplexing in multiuser
  {MIMO} channels,''
\newblock {\em IEEE Transactions on Signal Processing}, vol. 52, no. 2, pp. 461
  -- 471, 2004.

\bibitem{zhang09}
J.~Zhang, R.~Chen, J.G. Andrews, A.~Ghosh, and R.W. Heath,
\newblock ``Networked {MIMO} with clustered linear precoding,''
\newblock {\em IEEE Transactions on Wireless Communications}, , no. 8, pp.
  1910--1921, 2009.

\bibitem{zhang10JSAC}
R.~Zhang,
\newblock ``Cooperative multi-cell block diagonalization with per-base-station
  power constraints,''
\newblock {\em IEEE Journal on Selected Areas in Communications}, vol. 28, no.
  9, pp. 1435 --1445, 2010.

\bibitem{Laneman03}
J.N. Laneman and G.W. Wornell,
\newblock ``Distributed space-time-coded protocols for exploiting cooperative
  diversity in wireless networks,''
\newblock {\em IEEE Transactions on Information Theory}, vol. 49, no. 10, pp.
  2415 -- 2425, 2003.

\bibitem{Sendonaris03}
A.~Sendonaris, E.~Erkip, and B.~Aazhang,
\newblock ``User cooperation diversity. part i. system description,''
\newblock {\em IEEE Transactions on Communications}, vol. 51, no. 11, pp. 1927
  -- 1938, 2003.

\bibitem{Nassab08}
V.~Havary-Nassab, S.~Shahbazpanahi, A.~Grami, and Z.-Q. Luo,
\newblock ``Distributed beamforming for relay networks based on second-order
  statistics of the channel state information,''
\newblock {\em IEEE Transactions on Signal Processing}, vol. 56, no. 9, pp.
  4306 --4316, 2008.

\bibitem{Dehkordy09}
S.~Fazeli-Dehkordy, S.~Shahbazpanahi, and S.~Gazor,
\newblock ``Multiple peer-to-peer communications using a network of relays,''
\newblock {\em IEEE Transactions on Signal Processing}, vol. 57, no. 8, pp.
  3053 --3062, 2009.

\bibitem{Nam08}
S.~Nam, M.~Vu, and V.~Tarokh,
\newblock ``Relay selection methods for wireless cooperative communications,''
\newblock in {\em CISS 2008.}, march 2008, pp. 859 --864.

\bibitem{Zhao06}
Y.~Zhao, R.~Adve, and T.~Lim,
\newblock ``Improving amplify-and-forward relay networks: Optimal power
  allocation versus selection,''
\newblock in {\em 2006 IEEE International Symposium on Information Theory},
  july 2006, pp. 1234 --1238.

\bibitem{Love08}
D.~J. Love, R.~W. Heath, V.~K.~N. Lau, D.~Gesbert, B.D. Rao, and M.~Andrews,
\newblock ``An overview of limited feedback in wireless communication
  systems,''
\newblock {\em IEEE JSAC}, vol. 26, no. 8, pp. 1341 --1365, 2008.

\bibitem{huang13tsp}
Y.~Huang and B.~D. Rao,
\newblock ``An analytical framework for heterogeneous partial feedback design
  in heterogeneous multicell ofdma networks,''
\newblock {\em IEEE Transactions on Signal Processing}, vol. 61, no. 3, pp.
  753--769, 2013.

\bibitem{Ng10}
C.~T.~K. Ng and H.~Huang,
\newblock ``Linear precoding in cooperative {MIMO} cellular networks with
  limited coordination clusters,''
\newblock {\em IEEE Journal on Selected Areas in Communications}, vol. 28, no.
  9, pp. 1446 --1454, 2010.

\bibitem{Papadogiannis10}
A.~Papadogiannis and G.C. Alexandropoulos,
\newblock ``The value of dynamic clustering of base stations for future
  wireless networks,''
\newblock in {\em 2010 IEEE International Conference on Fuzzy Systems (FUZZ)},
  july 2010, pp. 1 --6.

\bibitem{hong12sparse}
M.~Hong, R.~Sun, H.~Baligh, and Z.-Q. Luo,
\newblock ``Joint base station clustering and beamformer design for partial
  coordinated transmission in heterogenous networks,''
\newblock 2012,
\newblock accepted by IEEE Journal on Selected Areas in Communications, Special
  issue on Large Scale Multi-antenna Systems.

\bibitem{kim12}
S.-J. Kim, S.~Jainand, and G.B. Giannakis,
\newblock ``Backhaul-constrained multi-cell cooperation using compressive
  sensing and spectral clustering,''
\newblock in {\em 2012 IEEE SPWAC}, 2012.

\bibitem{zeng10}
Y.~Zeng, E.~Gunawan, Y.~Guan, and J.~Liu,
\newblock ``Joint base station selection and linear precoding for cellular
  networks with multi-cell processing,''
\newblock in {\em IEEE TENCON}, nov. 2010, pp. 1976 --1981.

\bibitem{Ibrahim08}
A.~S. Ibrahim, A.~K. Sadek, W.~Su, and K.~J.~Ray Liu,
\newblock ``Cooperative communications with relay-selection: When to cooperate
  and whom to cooperate with?,''
\newblock {\em IEEE Transactions on Wireless Communications}, vol. 7, pp.
  2814--2827, 2008.

\bibitem{Jing09}
Y.~Jing and H.~Jafarkhani,
\newblock ``Single and multiple relay selection schemes and their achievable
  diversity orders,''
\newblock {\em IEEE Transactions on Wireless Communications}, vol. 8, no. 3,
  pp. 1414 --1423, 2009.

\bibitem{yu12relay}
C.-H. Yu and O.~Tirkkonen,
\newblock ``Opportunistic multiple relay selection with diverse mean channel
  gains,''
\newblock {\em IEEE Transactions on Wireless Communications}, vol. 11, no. 3,
  pp. 885 --891, 2012.

\bibitem{Michalopoulos06}
D.S. Michalopoulos, G.K. Karagiannidis, T.A. Tsiftsis, and R.K. Mallild,
\newblock ``Wlc41-1: An optimized user selection method for cooperative
  diversity systems,''
\newblock in {\em IEEE GLOBECOM '06.}, 27 2006-dec. 1 2006, pp. 1 --6.

\bibitem{garey79}
M.~R. Garey and D.~S. Johnson,
\newblock {\em Computers and Intractability: A guide to the Theory of
  {NP}-completeness},
\newblock W. H. Freeman and Company, San Francisco, U.S.A, 1979.

\bibitem{luo10SDPMagazine}
Z.-Q. Luo, W.-K. Ma, A.M.-C. So, Y.~Ye, and S.~Zhang,
\newblock ``Semidefinite relaxation of quadratic optimization problems,''
\newblock {\em IEEE Signal Processing Magazine}, vol. 27, no. 3, pp. 20 --34,
  2010.

\bibitem{Luo07approximationbounds}
Z.-Q. Luo, N.~D. Sidiropoulos, P.~Tseng, and S.~Zhang,
\newblock ``Approximation bounds for quadratic optimization with homogeneous
  quadratic constraints,''
\newblock {\em SIAM Journal on Optimization}, pp. 1--28, 2007.

\bibitem{Nemirovski_Roos_Terlaky_1999}
A.~Nemirovski, C.~Roos, and T.~Terlaky,
\newblock ``On maximization of quadratic form over intersection of ellipsoids
  with common center,''
\newblock {\em Mathematical Programming}, vol. 86, no. 3, pp. 463--473, 1999.

\bibitem{zhang:optimal11}
Y.~J. Zhang and A.~M.-C. So,
\newblock ``Optimal spectrum sharing in mimo cognitive radio networks via
  semidefinite programming,''
\newblock {\em IEEE Journal on Selected Areas in Communications}, pp. 362--373,
  2011.

\bibitem{Matskani08}
E.~Matskani, N.~Sidiropoulos, Z.-Q. Luo, and L.~Tassiulas,
\newblock ``Convex approximation techniques for joint multiuser downlink
  beamforming and admission control,''
\newblock {\em IEEE Transactions on Wireless Communications}, vol. 7, no. 7,
  pp. 2682 --2693, 2008.

\bibitem{Matskani09}
E.~Matskani, N.D. Sidiropoulos, Z.-Q. Luo, and L.~Tassiulas,
\newblock ``Efficient batch and adaptive approximation algorithms for joint
  multicast beamforming and admission control,''
\newblock {\em IEEE Transactions on Signal Processing}, vol. 57, no. 12, pp.
  4882 --4894, 2009.

\bibitem{liu13deflation}
Y.-F Liu, Y.-H. Dai, and Z.-Q. Luo,
\newblock ``Joint power and admission control via linear programming
  deflation,''
\newblock {\em IEEE Transactions on Signal Processing}, vol. 61, no. 6, pp.
  1327 --1338, 2013.

\bibitem{Goemans:1995}
M.~X. Goemans and D.~P. Williamson,
\newblock ``Improved approximation algorithms for maximum cut and
  satisfiability problems using semidefinite programming,''
\newblock {\em Journal of ACM}, vol. 42, no. 6, pp. 1115--1145, 1995.

\bibitem{hong12vmimo_proof}
M.~Hong, Z.~Xu, M.~Razaviyayn, and Z.-Q. Luo,
\newblock ``Joint user grouping and linear virtual beamforming: Complexity,
  algorithms and approximation bounds,''
\newblock 2012,
\newblock Online at:
  https://sites.google.com/site/mingyihong84/journal-publications.

\bibitem{moghaddam:spectral}
Baback Moghaddam, Yair Weiss, and Shai Avidan,
\newblock ``Spectral bounds for sparse pca: Exact and greedy algorithms,''
\newblock in {\em NIPS'05}, 2005.

\bibitem{marsch11clustering}
P.~Marsch and G.~Fettweis,
\newblock ``Static clustering for cooperative multi-point (comp) in mobile
  communications,''
\newblock in {\em 2011 IEEE International Conference on Communications (ICC)},
  june 2011, pp. 1 --6.

\bibitem{Marshall10}
A.~W. Marshall, I.~Olkin, and B.~C. Arnold,
\newblock {\em Inequalities: Theory of Majorization And Its Applications},
\newblock Springer, 2010.

\bibitem{greenLargeDeviation}
C.~J. Green,
\newblock ``Large deviation results for combinatorics and number theory,''
\newblock available at: www.dpmms.cam.ac.uk/~bjg23/papers/deviate.pdf.

\end{thebibliography}
}

\begin{IEEEbiography}[{\includegraphics[width=1in,height=1.25in,clip,keepaspectratio]
{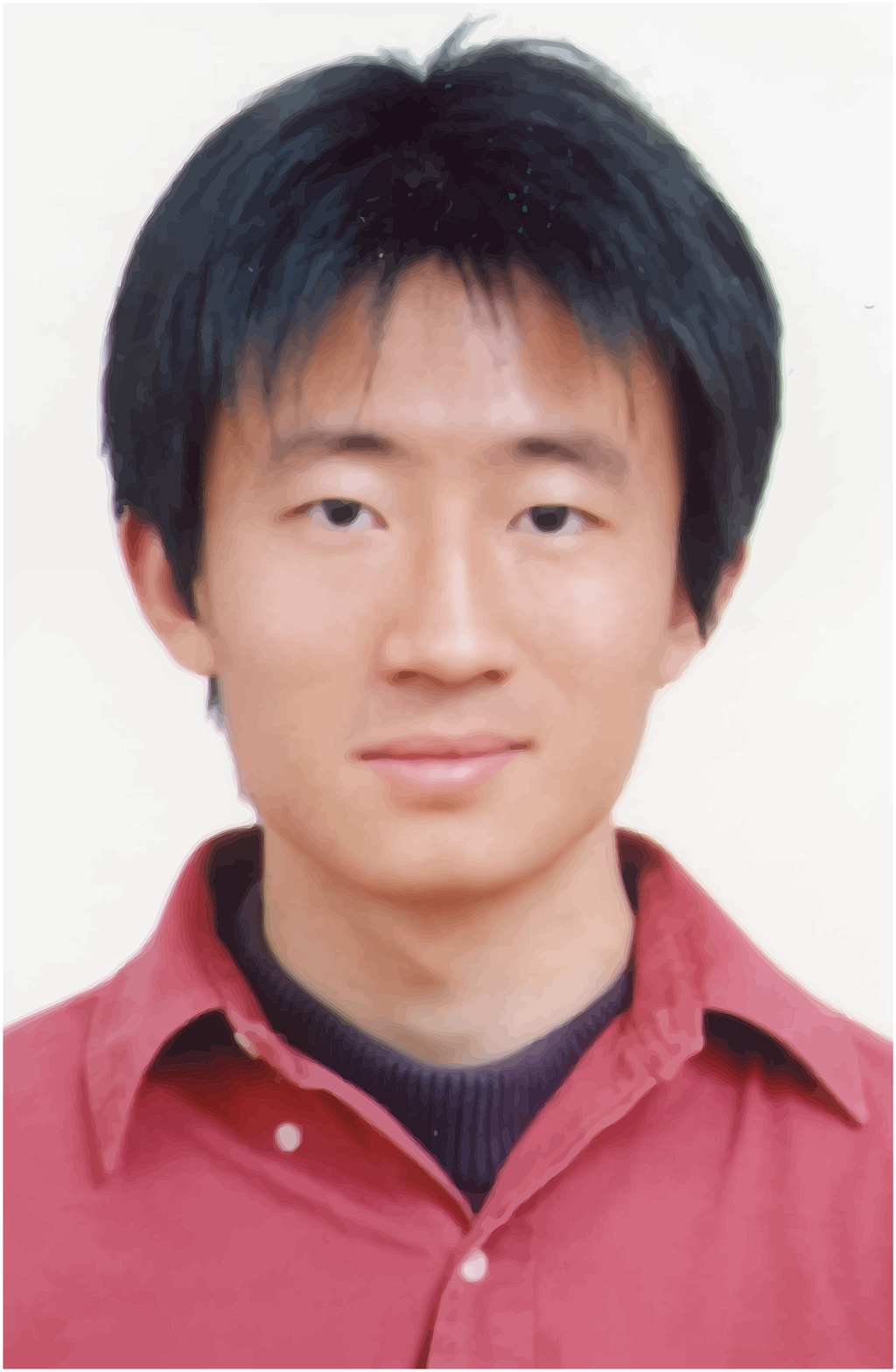}}]{Mingyi Hong} received his B.E. degree in
Communications Engineering from Zhejiang University, China, in 2005,
and his M.S. degree in Electrical Engineering from Stony Brook
University in 2007, and Ph.D. degree in Systems Engineering from
University of Virginia in 2011. He is currently a post-doctoral
fellow with the Department of Electrical and Computer Engineering,
University of Minnesota. His research interests are primarily in the
fields of statistical signal processing, wireless communications,
and optimization theory.
\end{IEEEbiography}

\begin{IEEEbiography}[{\includegraphics[width=1in,height=1.25in,clip,keepaspectratio]
{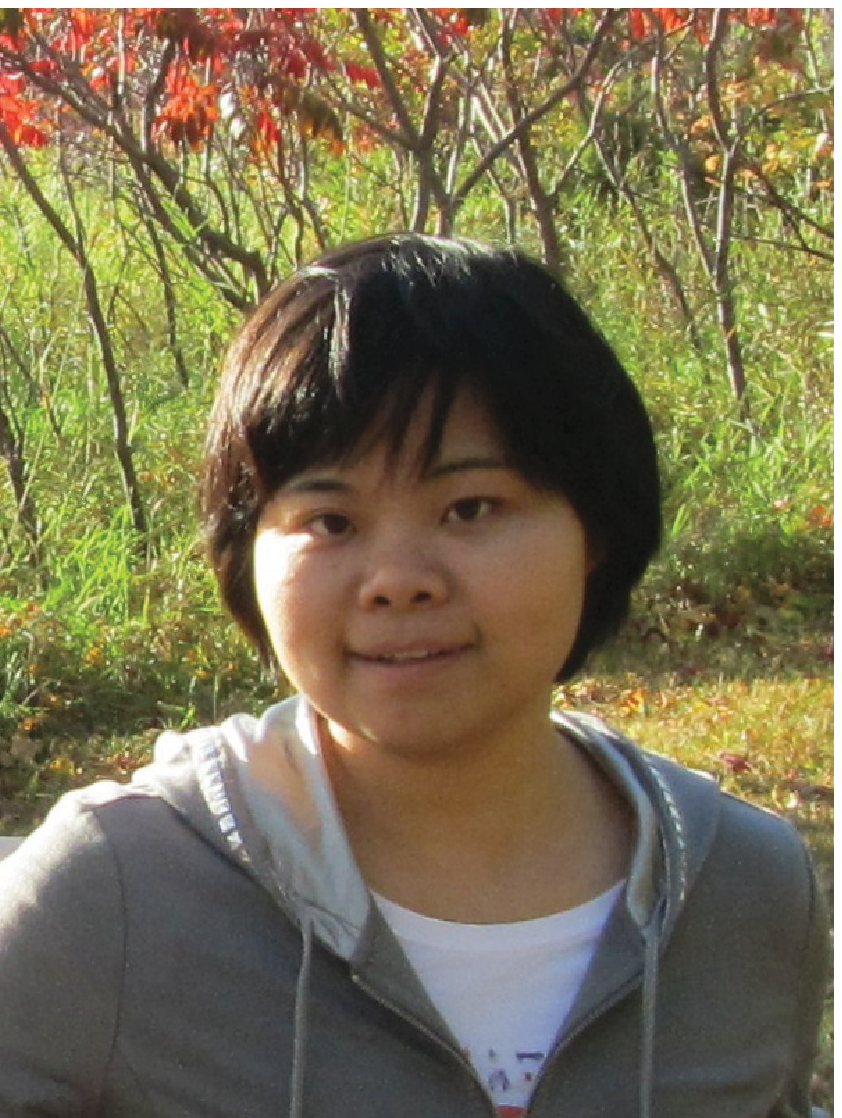}}]{Zi Xu} received the B.Sc. degree in applied
mathematics from the Hunan Normal University, Hunan, China, in 2003.
She then studied nonlinear optimization in the Institute of
Computational Mathematics and Scientific/Engineering Computing,
Chinese Academy of Sciences, and received the Ph.D. degree in
optimization in 2008. After her graduation, she has been with the
Department of Mathematics in Shanghai University, Shanghai, China,
and became an Associate Professor in 2012. Her research interests
include optimization algorithm, complexity analysis and various
optimization applications.

\end{IEEEbiography}

\begin{IEEEbiography}[{\includegraphics[width=1in,height=1.25in,keepaspectratio]
{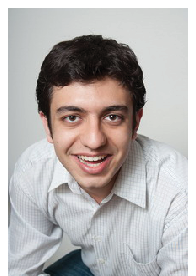}}] {Meisam Razaviyayn} received the B.Sc. degree
from Isfahan University of Technology, Isfahan, Iran, in 2008.
During summer 2010, he was working as a research intern at Huawei
Technologies. He is currently working towards the Ph.D. degree in
electrical engineering at the University of Minnesota. His research
interests include the design and analysis of efficient optimization
algorithms with application to data communication, signal
processing, and machine learning.
\end{IEEEbiography}

\begin{IEEEbiography}[{\includegraphics[width=1in,height=1.25in,clip,keepaspectratio]
{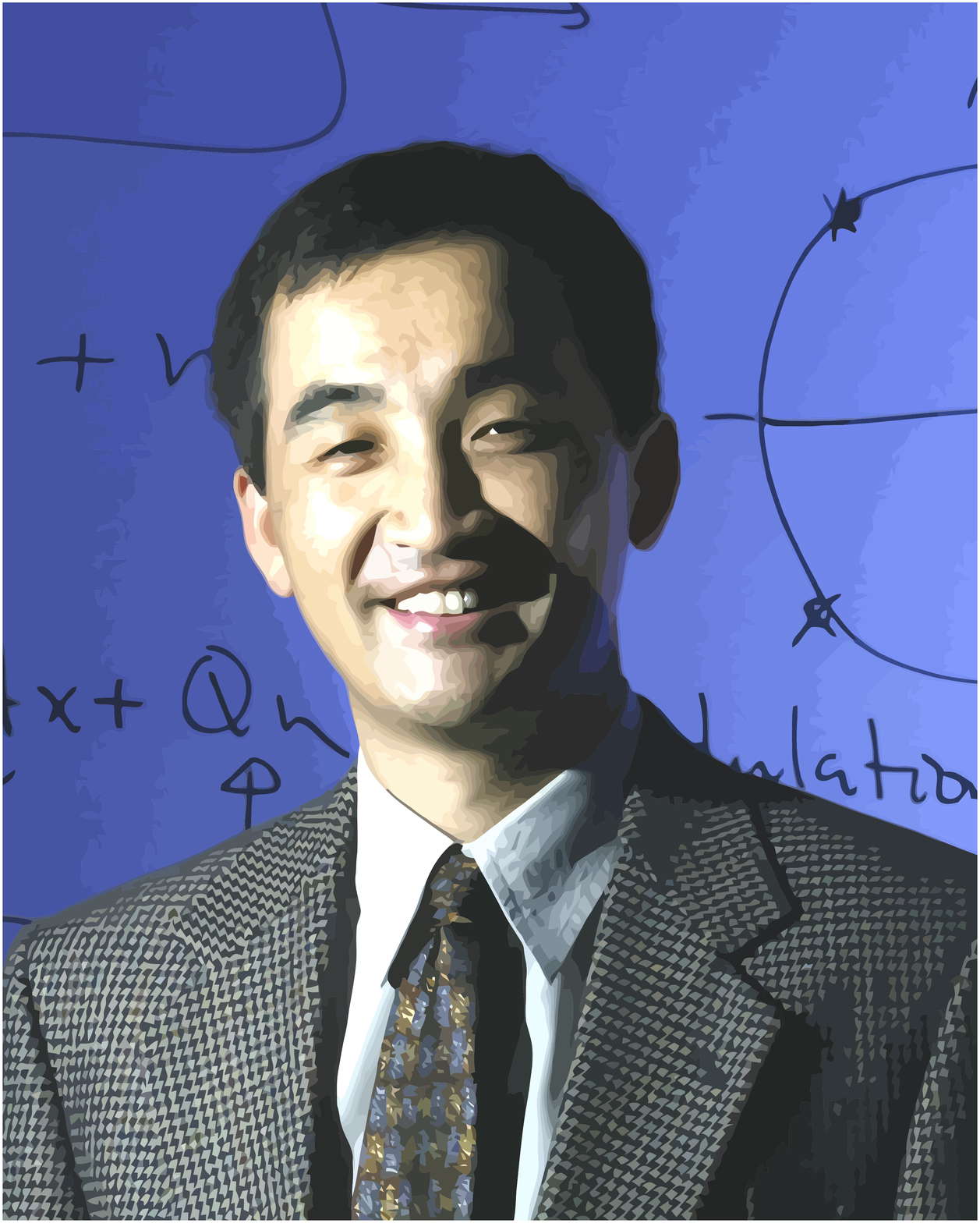}}]{Zhi-Quan Luo} received his B.Sc. degree in
Applied Mathematics in 1984 from Peking University, Beijing, China.
Subsequently, he was selected by a joint committee of the American
Mathematical Society and the Society of Industrial and Applied
Mathematics to pursue Ph.D study in the United States. After an
one-year intensive training in mathematics and English at the Nankai
Institute of Mathematics, Tianjin, China, he studied in the
Operations Research Center and the Department of Electrical
Engineering and Computer Science at MIT, where he received a Ph.D
degree in Operations Research in 1989. From 1989 to 2003, Dr. Luo
held a faculty position with the Department of Electrical and
Computer Engineering, McMaster University, Hamilton, Canada, where
he eventually became the department head and held a Canada Research
Chair in Information Processing. Since April of 2003, he has been
with the Department of Electrical and Computer Engineering at the
University of Minnesota (Twin Cities) as a full professor and holds
an endowed ADC Chair in digital technology. His research interests
include optimization algorithms, signal processing and digital
communication.

Dr. Luo is a fellow of IEEE and SIAM, and serves as the chair of the
IEEE Signal Processing Society Technical Committee on the Signal
Processing for Communications (SPCOM). He is a recipient of the 2004
and 2009 IEEE Signal Processing Society¡¯s Best Paper Awards, the
2010 Farkas Prize from the INFRMS Optimization Society, the 2011
EURASIP Best Paper Award and the 2011 ICC Best Paper Award. He has
held editorial positions for several international journals
including Journal of Optimization Theory and Applications, SIAM
Journal on Optimization, Mathematics of Computation, and IEEE
Transactions on Signal Processing. He currently serves as the
Editor-in-Chief for the journal IEEE Transactions on Signal
Processing.

\end{IEEEbiography}

\end{document}